\documentclass[11pt,reqno]{amsart}
\usepackage[english]{babel}
\usepackage{tikz,float}
\usepackage{enumerate}
\usepackage[active]{srcltx}
\theoremstyle{plain}
\newtheorem{theorem}{Theorem}[section]
\newtheorem{lemma}[theorem]{Lemma}
\newtheorem{proposition}[theorem]{Proposition}
\theoremstyle{remark}
\newtheorem{remark}[theorem]{Remark}
\theoremstyle{definition}
\newtheorem{assumption}{Assumption}

\newtheorem{definition}[theorem]{Definition}
\newcommand{\ZZ}{\mathbb{Z}}
\newcommand{\RR}{\mathbb{R}}
\newcommand{\CC}{\mathbb{C}}
\newcommand{\NN}{\mathbb{N}}
\newcommand{\EE}{\mathbb{E}}
\newcommand{\cL}{\mathcal{L}}

\newcommand{\cN}{\mathcal{N}}
\newcommand{\p}{\ensuremath{\mathbb{P}}}
\newcommand{\drm}{\ensuremath{\mathrm{d}}}
\newcommand{\abs}[1]{\left\vert#1\right\vert}
\newcommand{\sprod}[2]{\left\langle#1,#2 \right \rangle}
\newcommand{\norm}[1]{\left\Vert#1\right\Vert}
\newcommand{\Pro}{\ensuremath{P}}

\newcommand{\PP}{\mathbb{P}}

\newcommand{\euler}{\mathrm{e}}

\DeclareMathOperator{\supp}{\operatorname{supp}}
\DeclareMathOperator{\re}{Re}
\DeclareMathOperator{\im}{Im}
\DeclareMathOperator{\Tr}{Tr}
\DeclareMathOperator{\dist}{dist}
\DeclareMathOperator{\diam}{diam}
\renewcommand{\epsilon}{\varepsilon}
\renewcommand{\i}{\ensuremath{{\mathrm{i}}}}
\newcommand{\BIGOP}[1]{\mathop{\mathchoice%
{\raise-0.22em\hbox{\huge $#1$}}%
{\raise-0.05em\hbox{\Large $#1$}}{\hbox{\large $#1$}}{#1}}}
\newcommand{\BIGboxplus}{\mathop{\mathchoice%
{\raise-0.35em\hbox{\huge $\boxplus$}}%
{\raise-0.15em\hbox{\Large $\boxplus$}}{\hbox{\large $\boxplus$}}{\boxplus}}}
\newcommand{\bigtimes}{\BIGOP{\times}}
\begin{document}
%
%
%
%
%
%
%
%
\title[Fractional moment method for discrete alloy-type models]
{Anderson localization for a class of models with a sign-indefinite single-site potential via fractional moment method}
\author[A. Elgart]{Alexander Elgart}
\address{448 Department of Mathematics \\ McBryde Hall, Virginia Tech \\
Blacksburg, VA, 24061  \\ USA}
\email{aelgart@vt.edu}
\author[M. Tautenhahn]{Martin Tautenhahn}
\address{Technische Universit\"a{}t Chemnitz \\ Fakult\"a{}t f\"u{}r Mathematik \\ D-09107 Chemnitz \\ Germany}
\email{martin.tautenhahn@mathematik.tu-chemnitz.de}
\author[I. Veseli\'c]{Ivan Veseli\'c}
\address{Technische Universit\"a{}t Chemnitz \\ Fakult\"a{}t f\"u{}r Mathematik \\
 D-09107 Chemnitz \\ Germany}
\email{ivan.veselic@mathematik.tu-chemnitz.de}
\begin{abstract}
A technically convenient signature of Anderson localization is
exponential decay of the fractional moments of the Green function
within appropriate energy ranges. We consider a random Hamiltonian
on a lattice whose randomness is generated by the sign-indefinite
single-site potential, which is however sign-definite at the
boundary of its support. For this class of Anderson operators we
establish a finite-volume criterion which implies that above
mentioned the fractional moment decay property holds. This
constructive criterion is satisfied at typical perturbative regimes,
e.\,g. at spectral boundaries which satisfy 'Lifshitz tail
estimates' on the density of states and for sufficiently strong
disorder. We also show how the fractional moment method facilitates
the proof of exponential (spectral) localization for such random
potentials.
\end{abstract}
%
%
\subjclass{82B44, 60H25, 35J10}
\keywords{Fractional moment method, localization, discrete alloy-type model, non-monotone, sign-indefinite, single-site potential}
\maketitle
%
%
%
%
%
%
%
%
\section{Introduction}
The addition of disorder can have a profound effect on the spectral
and dynamical properties of a self adjoint differential operator. In
general terms, the effect is that in certain energy ranges the
absolutely continuous spectrum of the Laplacian that describes the
perfect crystal  may be modified to consist of a random dense set of
eigenvalues associated with localized eigenfunctions. Thus it
affects various properties of the corresponding model: time
evolution (non--spreading of wave packets), conductivity (in
response to electric field), and Hall currents (in the presence of
both magnetic and electric field). This phenomenon, known as
Anderson localization, was initially discussed in the context of the
conduction properties of metals, but the mechanism is of relevance
in a variety of other situations.
\par
The first breakthrough in understanding the spectral properties of
the multidimensional Anderson model is associated with the seminal
work of Fr\"ohlich and Spencer \cite{FroehlichS1983} that introduced
the method of the multiscale analysis (MSA). Ten years later,
Aizenman and Molchanov \cite{AizenmanM1993} realized how one can
greatly streamline the proof of the spectral localization for a
standard Anderson model, deriving the machinery of what is now known
as the fractional moment method (FMM). Both methods were
subsequently improved and generalized in a number of papers, {\it
e.g.}
\cite{FroehlichMSS1985,DreifusK1989,GerminetK2001,AizenmanFSH2001,BourgainK2005,
AizenmanENSS2006,AizenmanGKW2009}.
\par
The standard Anderson model assumes that the values of the random
potential at different sites of the lattice are uncorrelated.
Although it was realized early on that both MSA and FMM work well
for a more general class of models where some correlations are
permitted, the allowed randomness features the so called
monotonicity property. The simplest example of such correlated
randomness is an alloy-type model with a {\it fixed sign}
single-site potential of finite support. The hallmark of the
monotone alloy-type model is a regularity of the Green function under
the local averages which leads to the Wegner estimate used in MSA
and to the a-priori bound, used in FMM. Here the term a-priori bound
means that the average of an fractional power of the elements of the
Green's function is uniformly bounded.
\par
There is no physically compelling reason for a random tight binding
model to have such monotonicity property, and one can ask the
natural question whether the Anderson localization can be
established if one relinquishes it altogether. For alloy-type models
on the continuum with a sign-changing single-site potential
localization has been derived via MSA, e.\,g. in
\cite{Klopp1995,Veselic2002,KostrykinV2006,Klopp2002}, see also
\cite{Stolz2002}. All these results are build on recovering the
monotonicity, one way or another. The recent preprint of Kr\"uger
\cite{Krueger2010} establishes  Anderson localization for
non-monotone models on the lattice.
His proof relies on MSA and on the method of Bourgain \cite{Bourgain2009}
to obtain certain Wegner-like estimates.
The theorems of \cite{Krueger2010} address the strong disorder regime
only, although it is likely that they hold (as usual for the MSA method)
in all situations where an appropriate initial scale decay estimate for
the resolvent can be established.
In this paper we derive a general finite-volume criterion applicable to the strong disorder as well as Lifshitz tail regimes,
and apply it to establish the Anderson localization in the former one.
\par
In this paper we investigate how far one can push the FMM for the
non monotone discrete alloy-type model. In general, we don't expect
that the regularity of the Green function under the local averages
survives the complete relaxation of the monotonicity condition. It
turns out, however, that just monotonicity of the single-site
potential at the {\it boundary} of its support is sufficient to
initiate FMM. This condition allows us to combine monotone as well
as non-monotone techniques to establish the local a-priori bound of
the fractional moment of the Green function, which is a cornerstone
of FMM.
\par
As a consequence, we obtain a number of results that are parallel to
the ones established in the monotone case. In particular, we develop
 finite volume criterion: A set of certain conditions which when
satisfied by the alloy-type model obtained by restricting the full
operator to some finite volume are sufficient to deduce the
exponential decay of the typical Green function.
%
%
%
%
%
%
%
%
%
\section{Model and results} \label{sec:model}
Let $d \geq 1$. For $x \in \ZZ^d$ we recall the following standard norms $\lvert x \rvert_1 = \sum_{i=1}^d \lvert x_i \rvert$ and $\lvert x \rvert_\infty = \max\{\lvert x_1 \rvert,\ldots, \lvert x_d \rvert\}$. For $\Gamma \subset \ZZ^d$ we introduce the Hilbert space $\ell^2 (\Gamma) = \{\psi : \Gamma \to \CC : \sum_{k \in \Gamma} \lvert \psi (k) \rvert^2 < \infty\}$ with inner product $\sprod{\phi}{\psi} = \sum_{k \in \Gamma} \overline{\phi(k)} \psi (k)$. On $\ell^2 (\ZZ^d)$ we consider the discrete random Schr\"o{}dinger operator
\begin{equation} \label{eq:hamiltonian}
 H_\omega := -\Delta + \lambda V_\omega , \quad \lambda > 0 .
\end{equation}
Here, $\omega$ is an element of the probability space specified
below, $\Delta: \ell^2 \left(\ZZ^d\right) \to \ell^2
\left(\ZZ^d\right)$ denotes the discrete Laplace operator and
$V_\omega : \ell^2 \left(\ZZ^d\right) \to \ell^2 \left(\ZZ^d\right)$
is a random multiplication operator. They are defined by
\begin{equation*}
\left(\Delta \psi \right) (x) := \sum_{\abs{e}_1 = 1} \psi (x+e) \quad \mbox{and} \quad
\left( V_\omega \psi \right) (x) :=  V_\omega (x) \psi (x)
\end{equation*}
and represent the kinetic energy and the random potential energy, respectively. The parameter $\lambda$ models the strength of the disorder. We assume that the probability space has a product structure $\Omega :=
\bigtimes_{k \in \ZZ^d} \RR$ and is equipped with the probability measure $\p (\drm \omega) := \prod_{k \in \ZZ^d} \mu(\drm \omega_k)   $ where
$\mu$ is a probability measure on $\RR$.
Each element $\omega$ of $\Omega$ may be represented as a collection
$\{\omega_k\}_{k \in \ZZ^d}$ of real numbers, being the realization of a field of independent identically distributed (i.\,i.\,d.)
random variables, each distributed according to $\mu$. The symbol
$\mathbb{E}\{\cdot\}$ denotes the expectation with respect to the probability
measure, i.\,e. $\mathbb{E} \{\cdot\} := \int_\Omega (\cdot)
\p (\drm \omega)$. For a set $\Gamma \subset \ZZ^d$, $\mathbb{E}_\Gamma\{\cdot \}$
denotes the expectation with respect to $\omega_k$, $k \in \Gamma$. That is,
$\mathbb{E}_{\Gamma} \{\cdot\} := \int_{\Omega_\Gamma} (\cdot)
\prod_{k \in \Gamma} \mu(\drm \omega_k)$ where $\Omega_\Gamma
:= \bigtimes_{k \in \Gamma} \RR$. Let the \emph{single-site potential}
$u : \ZZ^d \to \RR$ be a function with finite and non-empty support $\Theta
:= \supp u = \{k \in \ZZ^d : u(k) \not = 0 \}$. We assume that the random
potential $V_\omega $ has an alloy-type structure, i.\,e. the potential value
\begin{equation*}
V_\omega (x) := \sum_{k \in \ZZ^d} \omega_k u (x-k)
\end{equation*}
at a lattice site $x \in \ZZ^d$ is a linear combination of the i.\,i.\,d. random
variables $\omega_k$, $k\in\ZZ^d$, with coefficients provided by the single-site
potential. For this reason we call the Hamiltonian \eqref{eq:hamiltonian} a \emph{discrete alloy-type model}. The function $u(\cdot - k)$ may be interpreted as a finite range potential associated to the lattice site $k\in\ZZ^d$. We assume (without loss of generality) that $0 \in \Theta$.
\par
Notice that the single-site potential $u$ may change its sign. As a consequence the quadratic form associated to $H_\omega$ does not necessarily depend in a monotone way on the random parameters $\omega_k$, $k \in \ZZ^d$. However, for our main result we have to assume that $u$ has fixed sign at the boundary of $\Theta$, see Assumption~\ref{ass:monotone}.
For $\Lambda \subset \ZZ^d$ we denote by $\partial^{\rm i} \Lambda = \{k \in \Lambda : \# \{j \in \Lambda : |k-j|_1 = 1\} < 2d\}$ the interior boundary of $\Lambda$ and by $\partial^{\rm o} \Lambda = \partial^{\rm i} \Lambda^{\rm c}$ the exterior boundary of $\Lambda$. Here $\Lambda^{\rm c} = \ZZ^d \setminus \Lambda$ denotes the complement of $\Lambda$.
\begin{assumption} \label{ass:monotone}\leavevmode \\[-2ex]
\begin{enumerate}
 \item [(A1)]
 The measure $\mu$ has a bounded, compactly supported density $\rho$.
\item [(A2)]
The function $u$ satisfies $u (k) > 0$ for all $k \in \partial^{\rm i} \Theta$.
\end{enumerate}
\end{assumption}
\begin{remark}
\begin{enumerate}[(i)]
 \item
This assumption plays an instrumental role in the proof of the
uniform boundedness of fractional moments of the Green's function
(a-priori bound), in the particular form presented in Lemma \ref{lemma:bounded},
and thus also of our main result, Theorem \ref{thm:result1}.

\item
Note that for models on $\ZZ$ Assumption (A2) can always be achieved by taking a linear combination
of several translates of the single site potential. With these linear combinations  one can work similarly as with the
original single site potential, cf.~Section 5 in \cite{ElgartTV2010}.
Actually, in the one-dimensional setting a particularly transparent version of
our proof is available:
The decoupling arguments of Section 4 in the present paper are replaced by
Lemma 3.3 of \cite{ElgartTV2010} which uses the special structure of the relevant resolvent matrix elements.

\item
For the purpose of comparison we present a different version of the a-priori bound
in the Appendix. It requires much milder conditions on $u$.
Unfortunatey, we do not see at the moment how it can be used to complete the proof of exponential
decay of fractional moments. See the Appendix for more details.
\end{enumerate}
\end{remark}

For the operator $H_\omega$ in \eqref{eq:hamiltonian} and $z \in \CC \setminus \sigma
(H_\omega)$ we define the corresponding \emph{resolvent} by $G_\omega (z)
= (H_\omega - z)^{-1}$. For the \emph{Green function}, which assigns
to each  $(x,y) \in \ZZ^d \times \ZZ^d$ the corresponding matrix element of the
resolvent, we use the notation
\begin{equation*} \label{eq:greens}
G_\omega (z;x,y) := \sprod{\delta_x}{(H_\omega - z)^{-1}\delta_y}.
\end{equation*}
For $\Gamma \subset \ZZ^d$, $\delta_k \in \ell^2 (\Gamma)$ denotes the
Dirac function given by $\delta_k (k) = 1$ for $k \in \Gamma$ and
$\delta_k (j) = 0$ for $j \in \Gamma \setminus \{k\}$.
Let $\Gamma_1 \subset \Gamma_2 \subset \ZZ^d$. We define the operator $\Pro_{\Gamma_1}^{\Gamma_2} : \ell^2 (\Gamma_2) \to \ell^2 (\Gamma_1)$ by
\[
 \Pro_{\Gamma_1}^{\Gamma_2} \psi := \sum_{k \in \Gamma_1} \psi (k) \delta_k .
\]
Note that the adjoint $(\Pro_{\Gamma_1}^{\Gamma_2})^* : \ell^2 (\Gamma_1) \to \ell^2 (\Gamma_2)$ is given by
\[
(\Pro_{\Gamma_1}^{\Gamma_2})^* \phi = \sum_{k \in \Gamma_1} \phi (k) \delta_k .
\]
If $\Gamma_2 = \ZZ^d$ we will drop the upper index and write $\Pro_{\Gamma_1}$ instead of $\Pro_{\Gamma_1}^{\ZZ^d}$.
For an arbitrary set $\Gamma \subset \ZZ^d$ we define the restricted operators $\Delta_\Gamma, V_\Gamma, H_\Gamma:\ell^2 (\Gamma) \to \ell^2 (\Gamma)$ by $\Delta_\Gamma := \Pro_\Gamma \Delta \Pro_\Gamma^\ast$, $V_\Gamma := \Pro_\Gamma V_\omega \Pro_\Gamma^\ast$ and
\[
 H_\Gamma := \Pro_\Gamma H_\omega \Pro_\Gamma^\ast = -\Delta_\Gamma + V_\Gamma .
\]
Furthermore, we define $G_\Gamma (z) := (H_\Gamma - z)^{-1}$ and $G_\Gamma (z;x,y) := \bigl\langle \delta_x, G_\Gamma (z) \delta_y \bigr\rangle$ for $z \in \CC \setminus \sigma (H_\Gamma)$ and $x,y \in \Gamma$. If $\Lambda \subset \ZZ^d$ is finite, $\lvert \Lambda \rvert$ denotes the number of elements of $\Lambda$.
\par
In order to formulate our main results, let us define the specific localization property we are interested in.
\begin{definition}
 Let $I \subset \RR$. A selfadjoint operator $H : \ell^2 (\ZZ^d) \to \ell^2 (\ZZ^d)$ is said to exhibit \emph{exponential localization in $I$}, if the spectrum of $H$ in $I$ is only of pure point type, i.\,e. $\sigma_{\rm c} (H) \cap I = \emptyset$, and the eigenfunctions of $H$ corresponding to the eigenvalues in $I$ decay exponentially. If $I=\RR$, we say that
$H$ exhibits \emph{exponential localization}.
\end{definition}
Our results are the following theorems.
\begin{theorem} \label{thm:result1}
Let $\Gamma \subset \ZZ^d$, $s \in (0,1/3)$ and suppose that Assumption \ref{ass:monotone} is satisfied.
Then for a sufficiently large $\lambda$ there are constants $\mu,A \in (0,\infty)$,
depending only on $d$, $\rho$, $u$,  $s$ and $\lambda$,
such that for all $z \in \CC \setminus \RR$ and all $x,y \in \Gamma$
\begin{equation*} \label{eq:result1}
\mathbb{E} \bigl\{\lvert G_\Gamma (z;x,y)\rvert^{s/(2\lvert \Theta \rvert)}\bigr\}\leq A \euler^{-\mu|x-y|_\infty} .
\end{equation*}
\end{theorem}

For $x \in \ZZ^d$ and $L>0$, we denote by $\Lambda_{L,x} = \{ k \in \ZZ^d : \lvert x-k  \rvert_\infty \leq L \}$
the cube of side length $2L+1$ centred at $x$.
\begin{theorem} \label{thm:result2}
Let $s \in (0,1)$, $C,\mu, \in (0,\infty)$, and $I \subset \RR$ be a interval. Assume that
\[
\EE \bigl\{ \lvert G_{\Lambda_{L,k}} (E + \i \epsilon;x,y) \rvert^{s} \bigr\} \leq C \euler^{-\mu \lvert x-y \rvert_\infty}
\]
for all $k \in \ZZ^d$, $L \in \NN$, $x,y \in \Lambda_{L,k}$, $E \in I$ and all $\epsilon \in (0,1]$.
Then $H_\omega$ exhibits exponential localization in $I$ for almost all $\omega \in \Omega$.
\end{theorem}
Let us emphasize that this result does not rely on Assumption \ref{ass:monotone}.
 Putting together Theorem~\ref{thm:result1} and Theorem~\ref{thm:result2}, we obtain exponential localization in the case of sufficiently large disorder.
\begin{theorem} \label{thm:result3}
Let Assumption \ref{ass:monotone} be satisfied and $\lambda$ sufficiently large.
Then $H_\omega$ exhibits exponential localization for almost all $\omega \in \Omega$.
\end{theorem}
Theorem~\ref{thm:result1} concerns the exponential decay of an
averaged fractional power of the Green function. It applies to
arbitrary finite $\Theta \subset \ZZ^d$ assuming that $u$ has fixed
sign on the interior vertex boundary of $\Theta$. In Section
\ref{sec:loc} we provide a new variant of the proof that the exponential decay of
an averaged fractional power of the Green function imply exponential
localization, which is formulated in Theorem \ref{thm:result2}.
\par
Theorem~\ref{thm:result1} and \ref{thm:result3} concern localization properties in the strong disorder regime.
We also prove a so called finite volume criterion, which can be used to establish exponential decay of an averaged fractional power of the Green function
at typical perturbative regimes.
 %
In particular, Theorem~\ref{thm:result1} follows from the finite volume criterion using the a-priori bound provided in Section~\ref{sec:bounded}.
\begin{theorem}[Finite volume criterion] \label{thm:result4}
Suppose that Assumption \ref{ass:monotone} is satisfied, let $\Gamma \subset \ZZ^d$, $z
\in \CC \setminus \RR$ with $\lvert z \rvert \leq m$ and $s \in (0,1/3)$. Then there
exists a constant $B_s$ which depends only on $d$, $\rho$, $u$, $m$, $s$,
such that if the condition
\begin{equation*}
b_s(\lambda, L,\Lambda): = \frac{B_s L^{3(d-1)} \Xi_s (\lambda)}{\lambda^{s/\lvert
\Theta \rvert}}\, \sum_{w\in\partial^{\rm o} W_x}\mathbb{E}
\bigl\{\lvert G_{\Lambda\setminus W_x} (z;x,w)\rvert^{s/(2\lvert
\Theta \rvert)}\bigr\}< b
\end{equation*}
is satisfied for some $b \in (0,1)$, arbitrary $\Lambda \subset \Gamma$, and all $x\in
\Lambda$, then for all $x,y \in \Gamma$
\begin{equation*}
\mathbb{E} \bigl\{\lvert G_\Gamma (z;x,y)\rvert^{s/(2\lvert \Theta
\rvert)} \bigr\}\leq A \euler^{-\mu|x-y|_\infty} .
\end{equation*}
Here
\[
A=\frac{C_s \Xi_s (\lambda)}{b} \quad \text{and} \quad
\mu=\frac{\lvert \ln b \rvert}{L+\diam \Theta + 2}\,,\] with $C_s$ inherited
from the a-priori bound (Lemma \ref{lemma:bounded}).
Here, the set $W_x$ is an certain annulus around $x$, defined precisely in Eq.~\eqref{eq:Wx} and the text below, $L \geq \diam \Theta + 2$ is some fixed number determining the size of the annulus $W_x$, and $\Xi_s (\lambda) = \max \{ \lambda^{- s/ 2 \lvert \Theta \rvert} , \lambda^{-2s} \}$.
\end{theorem}
\begin{remark}[Lifshitz-tail regime]
Apart from the strong disorder regime, a typical situation where the finite volume criterion can be verified are energies in a sufficiently small
neighbourhood of a
fluctuation boundary of the spectrum.

By this we mean that there is an energy $ E_0 \in \RR$, a neighbourhood size $\epsilon_0 >0$, and a diameter scaling exponent $ D \in \NN$
such that for any power $k \in \NN$ there exists a finite $ C_k  \in \NN$ and a scale $L_0\in \NN$, such that
\begin{equation}
 \label{eq:Lifshitz} \forall \ L > L_0, \epsilon \in (0,\epsilon_0)  \ : \
\PP\left \{  \omega  \mid \dist(\sigma(H_{\Lambda_{L,0}}), E_0) <\epsilon\right \} \le C_k \ \epsilon^k  \ L^D.
\end{equation}
In this situation one can use the a-priori bound in Lemma
\ref{lemma:bounded} and Combes-Thomas bound \cite {CombesT-73} along the lines of the
argument carried out in Section 5 of \cite{AizenmanENSS2006} to
establish the hypothesis of Theorem \ref{thm:result4}. Here a few
more comments are in order, since our model does not satisfy the
stochastic regularity assumptions on the random potential required
in \cite{AizenmanENSS2006}. Combes-Thomas estimates are
deterministic in nature, thus they remain unaffected by this change.
The mentioned regularity assumptions are needed to make sure that an
a-priori bound holds and that potential values at large distances
are independent. These two facts hold (for other reasons) for models
considered here. (Also, for our finite volume criterion one needs a
larger value of $\xi$ compared to Theorem 5.3 in
\cite{AizenmanENSS2006}. This is no obstacle since in the
Lifshitz-tail regime one can choose arbitrarily large $\xi$, by
taking the power $k$ in  \eqref{eq:Lifshitz} appropriately large.)
\end{remark}

Our paper is organized as follows. In Section~\ref{sec:bounded} we
show the boundedness of an averaged fractional power of the Green
function, which is an important ingredient of the finite volume
criterion proven in Section~\ref{sec:exp_decay}. In Section~\ref{sec:exp_decay}
we prove the finite volume criterion and Theorem~\ref{thm:result1} which
follows from the a-priori bound and the finite volume criterion. In
Section~\ref{sec:loc} we establish Theorems ~\ref{thm:result2} and
\ref{thm:result3}.
%
%
%
%
%
%
%
%
\section{Boundedness of fractional moments} \label{sec:bounded}
In this section we prove the boundedness of an averaged fractional power of the Green function. The right hand side of the estimate depends in a quantitative way on the disorder. In particular it implies that the bound gets small in the high disorder regime. The estimate on the fractional moment of the Green function is used iteratively in the next section, where we prove exponential decay of the Green function.
\par
In this section we consider the situation when Assumption \ref{ass:monotone} holds.
Let us define $R = \max \{ \lvert \inf \supp \rho \rvert , \lvert \sup \supp \rho \rvert \}$ where $\rho$ is the density of $\mu$.
Our main result of this section is Lemma~\ref{lemma:bounded}. In the proof we will use several lemmata whose formulation is postponed to the second part of this section.
\par
First, let us introduce some more notation. For $x \in \ZZ^d$ we denote by $\cN (x) = \{k \in \ZZ^d : |x-k|_1 = 1\}$ the neighborhood of $x$. For $\Lambda \subset \ZZ^d$ and $x \in \ZZ^d$ we define $\Lambda^+ = \Lambda \cup \partial^{\rm o} \Lambda$, $\Lambda_x = \Lambda + x = \{k \in \ZZ^d : k-x \in \Lambda\}$ and $u_{\rm min}^\Lambda = \min_{k \in \Lambda} \lvert u(k) \rvert$.
\begin{lemma}[A-priori bound] \label{lemma:bounded}
Let Assumption~\ref{ass:monotone} be satisfied, $\Gamma \subset \ZZ^d$, $m > 0$ and $s \in (0,1)$.
\begin{enumerate}[(a)]
 \item Then there is a constant $C_s$, depending only on $d$, $\rho$, $u$, $m$ and $s$, such that for all $z \in \CC \setminus \RR$ with $|z| \leq m$, all $x,y \in \Gamma$ and all $b_x,b_y \in \ZZ^d$ with $x \in \Theta_{b_x}$ and $y \in \Theta_{b_y}$
\[
\EE_{N} \Bigl\{ \bigl\lvert G_\Gamma (z;x,y) \bigr\rvert^{s/(2|\Theta|)} \Bigr\} \leq C_s \Xi_s (\lambda),
\]
where $\Xi_s (\lambda) = \max \{ \lambda^{- s/(2 \lvert \Theta \rvert)} , \lambda^{-2s} \}$ and $N = \{b_x , b_y\} \cup \cN (b_x) \cup \cN (b_y)$.
\item Then there is a constant $D_s$, depending only on $d$, $\rho$, $u$ and $s$, such that for all $z \in \CC \setminus \RR$, all $x,y \in \Gamma$ and all $b_x,b_y \in \ZZ^d$ with
\[
x \in \Theta_{b_x} \cap \Gamma \subset \partial^{\rm i} \Theta_{b_x} \quad \text{and} \quad
y \in \Theta_{b_y} \cap \Gamma \subset \partial^{\rm i} \Theta_{b_y}
\]
we have
\[
 \EE_{\{b_x , b_y\}} \Bigl\{ \bigl\lvert G_\Gamma (z;x,y) \bigr\rvert^s \Bigr\} \leq D_s \lambda^{-s} .
\]
\end{enumerate}
\end{lemma}
\begin{proof}
First we prove (a). Fix $x,y\in\Gamma$ and choose $b_x, b_y \in \ZZ^d$ in such a way that $x \in \Theta_{b_x}$ and $y \in \Theta_{b_y}$. This is always possible, and sometimes even with a choice $b_x = b_y$. However, we assume $b_x \not = b_y$. The case $b_x = b_y$ is similar but easier. Let us note that $\Theta_{b_x}$ and $\Theta_{b_y}$ are not necessarily disjoint. We apply Lemma~\ref{lemma:schur2} with $\Lambda_1 = \Theta_{b_x} \cup \Theta_{b_y} \cap \Gamma$ and $\Lambda_2 = \Lambda_1^+ \cap \Gamma$ and obtain
\begin{equation} \label{eq:lemma:bounded1}
\Pro_{\Lambda_1}^{\Gamma}(H_\Gamma - z)^{-1} (\Pro_{\Lambda_1}^{\Gamma})^* =
\bigl( H_{\Lambda_1} - z + \Pro_{\Lambda_1} \Delta \Pro_{\partial^{\rm o} \Lambda_1}^* (K-z)^{-1} \Pro_{\partial^{\rm o} \Lambda_1} \Delta \Pro_{\Lambda_1}^* \bigr)^{-1}
\end{equation}
where
\[
 K = H_{\partial^{\rm o} \Lambda_1} - \Pro_{\partial^{\rm o} \Lambda_1}^{\Lambda_1^+} B_\Gamma^{\Lambda_1^+} (\Pro_{\partial^{\rm o} \Lambda_1}^{\Lambda_1^+})^* .
\]
We note that $B_\Gamma^{\Lambda_1^+}$ depends only on the potential values $V_\omega (k)$, $k \in \Gamma \setminus \Lambda_1^+$ and is hence independent of $\omega_k$, $k \in \{b_x , b_y\} \cup \cN (b_x) \cup \cN (b_y)$. We also note that $K$ is independent of $\omega_{b_x}$ and $\omega_{b_y}$, and that the potential values $V_\omega (k)$, $k \in \partial^{\rm o} \Lambda_1$ depend monotonically on $\omega_k$, $k \in \cN (b_x) \cup \cN (b_y) =: N'$, by Assumption~\ref{ass:monotone}. More precisely, we can decompose $K : \ell^2 (\partial^{\rm o} \Lambda_1) \to \ell^2 (\partial^{\rm o} \Lambda_1)$ according to
\[
 K = A + \lambda \sum_{k \in N'} \omega_k V_k
\]
with some $A,V_k : \ell^2 (\partial^{\rm o} \Lambda_1) \to \ell^2 (\partial^{\rm o} \Lambda_1)$ and the properties that $A$ is independent of $\omega_k$, $k \in N'$, and $V := \sum_{k \in N'} V_k$ is diagonal and strictly positive definite with $V \geq u_{\rm min}^{\partial^{\rm i} \Theta}$. We fix $v \in N'$ and obtain with the transformation $\omega_v = \zeta_v$ and $\omega_i = \zeta_v + \zeta_i$ for $i \in N' \setminus \{v\}$ for all $t \in (0,1)$
\begin{align}
\EE_{N'} \Bigl\{ \bigl\lVert (K-z)^{-1} \bigr\rVert^t \Bigr\}
& = \!\!\!\!\!\!\!\!\! \int\limits_{[-R , R]^{\lvert N' \rvert}} \!\!\!\!\! \bigl\lVert (A - z + \lambda \sum_{k \in N'} \omega_k V_k)^{-1} \bigr\rVert^t \prod_{k \in N'} \rho (\omega_k) \drm \omega_k \nonumber \\[1ex]
& \leq \lVert \rho \rVert_\infty^{\lvert N' \rvert - 1} \!\!\!\!\!\!\!\!\! \int\limits_{[-S,S]^{\lvert N' \rvert}} \!\!\!\! \bigl\lVert (\tilde A + \zeta_v \lambda V)^{-1} \bigr\rVert^t \rho(\zeta_v) \drm \zeta_v \!\!\!\! \prod_{i \in N' \setminus \{v\}} \!\!\!\! \drm \zeta_i \label{eq:K}
\end{align}
where $S = 2R$ and $\tilde A = A - z + \lambda \sum_{k \in N' \setminus \{v\}} \zeta_i V_i$. The monotone spectral averaging estimate in Lemma~\ref{lemma:monotone2} gives for $t \in (0,1)$
\begin{equation*}
\EE_{N'} \Bigl\{ \bigl\lVert (K-z)^{-1} \bigr\rVert^t \Bigr\} \leq
\frac{\lVert \rho \rVert_\infty^{|N'|-1} (4R)^{|N'|-1} (C_{\rm W} \lvert \partial^{\rm o} \Lambda_1 \rvert \lVert \rho \rVert_\infty)^t}{(u_{\rm min}^{\partial^{\rm i} \Theta} \lambda)^t (1-t)} .
\end{equation*}
Hence there is a constant $C_1 (t)$ depending only on $\rho$, $u$, $d$, $\Lambda_1$ and $t$, such that
\begin{equation}  \label{eq:lemma:bounded2}
 \EE_{N'} \Bigl\{ \bigl\lVert (K-z)^{-1} \bigr\rVert^t \Bigr\} \leq \frac{C_1 (t)}{\lambda^t} .
\end{equation}
We use the notation $u_j$ for the translates of $u$, i.\,e. $u_j (x) = u(x-j)$ for all $j,x \in \ZZ^d$, as well as for the corresponding multiplication operator. The operator $H_{\Lambda_1} = -\Delta_{\Lambda_1} + V_{\Lambda_1}$ can be decomposed in $H_{\Lambda_1} = \tilde A' + \lambda \omega_{b_x} V_x + \lambda \omega_{b_y}  V_y$, where the multiplication operators $V_x,V_y : \ell^2 (\Lambda_1) \to \ell^2 (\Lambda_1)$ are given by $V_x (k) = u_{b_x} (k)$ and $V_y (k) = u_{b_y} (k)$, and where $\tilde A' = -\Delta_{\Lambda_1} + \lambda\sum_{k \in \ZZ^d \setminus \{b_x,b_y\}} \omega_k u_k$. Notice that $V_x$ is invertible on $\Theta_{b_x}$ and $V_y$ is invertible on $\Theta_{b_y}$. Hence there exists an $\alpha \in (0,1]$ such that $V_x + \alpha V_y$ is invertible on $\Lambda_1$. By Eq. \eqref{eq:lemma:bounded1} and this decomposition we have for all $t \in (0,1)$
\begin{align*}
 E &:= \EE_{\{b_x,b_y\}} \Bigl\{\bigl\lVert \Pro_{\Lambda_1}^{\Gamma}(H_\Gamma - z)^{-1} (\Pro_{\Lambda_1}^{\Gamma})^* \bigr\rVert^{t / |\Lambda_1|} \Bigr\} \\[1ex]
&= \int_{-R}^R \int_{-R}^R \bigl\lVert ( A' + \lambda \omega_{b_x} V_x + \lambda \omega_{b_y} V_y )^{-1} \bigr\rVert^{t/|\Lambda_1|} \rho(\omega_{b_x})\rho(\omega_{b_y}) \drm \omega_{b_x} \drm \omega_{b_y} ,
\end{align*}
where
\[
 A' = \tilde A' - z + \Pro_{\Lambda_1} \Delta \Pro_{\partial^{\rm o} \Lambda_1}^* (K-z)^{-1} \Pro_{\partial^{\rm o} \Lambda_1} \Delta \Pro_{\Lambda_1}^* .
\]
Notice that $\tilde A'$ and $K$ are independent of $\omega_{b_x}$ and $\omega_{b_y}$. Set $V := V_x + \alpha V_y$. We use the transformation $\omega_{b_x} = \zeta_{x}$, $\omega_{b_y} = \alpha \zeta_x + \zeta_y$ and obtain by Lemma~\ref{lemma:averagenorm}
\begin{align*}
 E &\leq \lVert \rho \rVert_\infty \int_{-2R}^{2R} \int_{-2R}^{2R} \bigl\lVert (A' + \zeta_y \lambda V_y + \zeta_x \lambda V)^{-1} \bigr\rVert^{t/|\Lambda_1|} \rho (\zeta_x) \drm \zeta_x \drm \zeta_y \\[1ex]
& \leq \lVert \rho \rVert_\infty \int_{-2R}^{2R}
  \frac{\lVert \rho \rVert_\infty^t \bigl( \lVert A' + \zeta_y \lambda V_y \rVert + 2R \lambda \lVert V \rVert \bigr)^{t(|\Lambda_1| - 1)/|\Lambda_1|}}
  {t^t 2^{-t} (1-t) \lambda^t \lvert \det V \rvert^{t/|\Lambda_1|}} \drm \zeta_y \\[1ex]
& \leq
  \frac{4R \lVert \rho \rVert_\infty^{t+1} \bigl( \lVert A' \rVert + 2R \lambda \lVert V_y \rVert + 2R \lambda \lVert V \rVert \bigr)^{t(|\Lambda_1| - 1)/|\Lambda_1|}}
  {t^t 2^{-t} (1-t) \lambda^t \lvert \det V \rvert^{t/|\Lambda_1|}} .
\end{align*}
The norm of $A'$ can be estimated as
\[
 \lVert  A' \rVert \leq 2d + (|\Theta| - 1) \lVert u \rVert_\infty + m + (2d)^2 \lVert (K-z)^{-1} \rVert .
\]
For the norm of $V_y$ and $V$ we have $\lVert V_y \rVert \leq \lVert u \rVert_\infty$ and $\lVert V \rVert \leq 2 \lVert u \rVert_\infty$. To estimate the determinant of $V$ we set $v_i = (u (i - b_x) , u(i - b_y))^{\rm T} \in \RR^2$ for $i \in \Lambda_1$, and $r = (1,\alpha)^{\rm T} \in \RR^2$. Then,
\[
 \lvert \det V \rvert =  \prod_{i \in \Lambda_1} \bigl\lvert u(i-b_x) + \alpha u(i-b_y) \bigr\rvert = \prod_{i \in \Lambda_1} \lVert v_i \rVert \bigl\lvert \langle r , v_i / \lVert v_i \rVert \rangle \bigr\rvert .
\]
Since we can choose $\alpha \in (0,1]$ in such a way that the distance of $r$ to each hyperplane $H_i = \{x_1,x_2 \in \RR : u(i-b_x)x_1 + u(i-b_y) x_2 = 0\}$, $i \in \Lambda_1$, is at least $d_0 = \sqrt{2} / (4(\lvert \Lambda_1 \rvert + 1))$, we conclude using $\lVert v_i \rVert \geq \sqrt{2} u_{\rm \min}^\Theta$
\[
  \lvert \det V \rvert \geq \prod_{i \in \Lambda_1} \lVert v_i \rVert d_0 \geq \left( \frac{u_{\rm \min}^\Theta}{2(\lvert \Lambda_1 \rvert + 1)} \right)^{\lvert \Lambda_1 \rvert} .
\]
Putting all together we see that there are constants $C_2 (t)$, $C_3 (t)$ and $C_4 (t)$ depending only on $\rho$, $u$, $d$, $m$, $\Lambda_1$ and $t$, such that
\begin{equation} \label{eq:withoutA}
 E \leq \frac{C_2(t)}{\lambda^t} + \frac{C_3(t)}{\lambda^{t/\lvert \Lambda_1 \rvert}} + \frac{C_4(t)}{\lambda^t} \lVert (K-z)^{-1} \rVert^{t \frac{|\Lambda_1| - 1}{|\Lambda_1|}} .
\end{equation}
If we average with respect to $\omega_k$, $k \in \cN(b_x)\cup \cN (b_y)$ we obtain by Eq.~\eqref{eq:lemma:bounded2}
\[
  \EE_{\cN(b_x)\cup \cN (b_y)} \bigl\{ E \bigr\} \leq \frac{C_2(t)}{\lambda^t} + \frac{C_3 (t)}{\lambda^{t/\lVert \Lambda_1 \rVert}} + \frac{C_4 (t) C_1 (t(|\Lambda_1| - 1)/|\Lambda_1|)}{\lambda^t \lambda^{t(|\Lambda_1| - 1)/|\Lambda_1|}} .
\]
Notice that $1 \leq \lvert \Lambda_1 \rvert \leq 2 \lvert \Theta
\rvert$. Now we choose $t = s|\Lambda_1|/(2|\Theta|)$ and eliminate
$\Lambda_1$ from the constants $C_1(t)$, $C_2(t)$, $C_3(t)$ and
$C_4(t)$ by maximizing them with respect to $\lvert \Lambda_1 \rvert
\in \{1, \dots , 2 \lvert \Theta \rvert\}$. We obtain that there are
constants $\tilde C_1 (s) , \tilde C_2 (s)$ and $\tilde C_3 (s)$,
depending only on $\rho$, $u$, $d$, $m$, and $s$, such that
\begin{align*}
 \EE_N \Bigl\{\bigl\lVert \Pro_{\Lambda_1}^{\Gamma}(H_\Gamma - z)^{-1} (\Pro_{\Lambda_1}^{\Gamma})^* \bigr\rVert^{\frac{s}{2\lvert \Theta \rvert}} \Bigr\}
& \leq
\frac{\tilde C_1 (s)}{\lambda^{s\frac{\lvert \Lambda_1 \rvert}{2\lvert \Theta \rvert}}}+
\frac{\tilde C_2 (s)}{\lambda^{\frac{s}{2\lvert \Theta \rvert}}}+
\frac{\tilde C_3 (s)}{\lambda^{s \frac{2\lvert \Lambda_1 \rvert - 1}{2 \lvert \Theta \rvert}}} \\[1ex]
& \leq
(\tilde C_1 (s)+\tilde C_2 (s)+\tilde C_3 (s)) \Xi_s (\lambda) .
\end{align*}
In the last estimate we have distinguished the cases $\lambda \geq 1$ and $\lambda < 1$ and used the fact that $1 \leq \lvert \Lambda_1 \rvert \leq 2 \lvert \Theta
\rvert$. This completes the proof of part (a).
\par
To prove (b) we fix $x,y \in \Gamma$ and $b_x,b_y \in \ZZ^d$ with $x \in \Theta_{b_x} \cap \Gamma \subset \partial^{\rm i} \Theta_{b_x}$ and $y \in \Theta_{b_y} \cap \Gamma \subset \partial^{\rm i} \Theta_{b_y}$. We again assume $b_x \not = b_y$. The case $b_x = b_y$ is similar but easier. We apply Lemma \ref{lemma:schur1} with $\Lambda = (\Theta_{b_x} \cup \Theta_{b_y})\cap \Gamma$ and obtain
\[
\Pro_\Lambda^\Gamma (H_\Gamma - z)^{-1} (\Pro_\Lambda^\Gamma)^* =
(H_\Lambda - B_\Gamma^\Lambda - z)^{-1}.
\]
Notice that $B_\Gamma^\Lambda$ is independent of $\omega_k$, $k \in \{b_x,b_y\}$.
By assumption, the potential values in $\Lambda$ depend \textit{monotonically} on $\omega_{b_x}$ and $\omega_{b_y}$. More precisely, we can rewrite the potential in the form $V_\Lambda = A + \omega_{b_x} \lambda V_x + \omega_{b_y} \lambda V_y$ with the properties that $A$ is independent of $\omega_k$, $k \in \{b_x,b_y\}$, and $V = V_x + V_y$ is strictly positive definite with $V \geq u_{\rm min}^{\partial^{\rm i} \Theta}$. We proceed similarly as in Ineq. \eqref{eq:K} and obtain using Lemma~\ref{lemma:monotone}
\[
 \EE_{\{b_x,b_y\}} \Bigl\{ \bigl\lVert \Pro_\Lambda^\Gamma (H_\Gamma - z)^{-1} (\Pro_\Lambda^\Gamma)^* \bigr\rVert^s \Bigr\} \leq \lVert \rho \rVert_\infty 4R \frac{(\lvert \Lambda \rvert u_{\rm min}^{\partial^{\rm i} \Theta} \lVert \rho \rVert_\infty)^s}{\lambda^s (1-s)} .
\]
We estimate $\lvert \Lambda \rvert \leq 2 \lvert \Theta \rvert$ and obtain part (b).
\end{proof}
\begin{remark}
 Note that even if Assuption \ref{ass:monotone} is not satisfied we obtain the bound \eqref{eq:withoutA}, namely
\begin{multline*}
 \EE_{\{b_x,b_y\}} \Bigl\{\bigl\lVert \Pro_{\Lambda_1}^{\Gamma}(H_\Gamma - z)^{-1} (\Pro_{\Lambda_1}^{\Gamma})^* \bigr\rVert^{t / |\Lambda_1|} \Bigr\} \\ \leq \frac{C_2(t)}{\lambda^t} + \frac{C_3(t)}{\lambda^{t/\lvert \Lambda_1 \rvert}} + \frac{C_4(t)}{\lambda^t} \lVert (K-z)^{-1} \rVert^{t \frac{|\Lambda_1| - 1}{|\Lambda_1|}} .
\end{multline*}
\end{remark}
Next we state and prove the tools used in the proof of Lemma \ref{lemma:bounded}.
The first set of these auxiliary results concerns spectral averaging, both in the monotone and in the non-monotone case. We start with an averaging lemma for determinants.
\begin{lemma} \label{lemma:det}
 Let $n \in \NN$ and $A, V \in \CC^{n \times n}$ be two matrices and assume that $V$ is invertible. Let further $0 \leq \rho \in L^1(\RR) \cap L^\infty (\RR)$ and $s \in (0,1)$. Then we have for all $\lambda > 0$ the bound
\begin{align}
 \int_{\RR} \abs{\det (A + rV)}^{-s/n} \rho (r) \drm r
&\leq \abs{\det V}^{-s/n} \Vert \rho \Vert_{L^1}^{1-s} \Vert\rho\Vert_{\infty}^{s} \frac{2^{s} s^{-s}}{1-s} \label{eq:det1} \\[1ex]
&\leq \abs{\det V}^{-s/n}\Bigl( \lambda^{-s} \Vert\rho\Vert_{L^1} + \frac{2 \lambda^{1-s}}{1-s} \Vert\rho\Vert_\infty  \Bigr) \label{eq:det2} .
\end{align}
\end{lemma}
\begin{proof}
 Since $V$ is invertible, the function $r \mapsto \det (A + rV)$ is a polynomial of order $n$ and thus the set $\{r \in \mathbb{R} \colon A + rV \text{ is singular}\}$ is a discrete subset of $\mathbb{R}$ with Lebesgue measure zero. We denote the roots of the polynomial by $z_1,\dots , z_n \in \CC$. By multilinearity of the determinant we have
\[
 \abs{\det (A + rV)} = \abs{\det V} \prod_{j=1}^n |r - z_j| \geq
 \abs{\det V} \prod_{j=1}^n |r - \re{z_j}| .
\]
The H\"older inequality implies for $s \in (0,1)$ that
\begin{equation*}
 \int_\RR \abs{\det (A + rV)}^{-s/n} \rho (r) \drm r \leq \abs{\det V}^{-s/n} \prod_{j = 1}^n \left( \int_\RR |r - \re z_j|^{-s} \rho (r) \drm r  \right)^{1/n} .
\end{equation*}
For arbitrary $\lambda > 0$ and all $z \in \RR$ we have
\begin{align*}
\int_\RR \frac{1}{\abs{r - z}^{s}} \rho(r)\drm r &=  \int\limits_{\abs{r - z} \geq \lambda} \frac{1}{\abs{r - z}^{s}} \rho(r)\drm r + \int\limits_{\abs{r - z} \leq \lambda} \frac{1}{\abs{r - z}^{s}} \rho(r)\drm r \\[1ex]
& \leq \lambda^{-s} \Vert\rho\Vert_{L^1} + \Vert\rho\Vert_\infty \frac{2 \lambda^{1-s}}{1-s}
\end{align*}
which gives Ineq.{} \eqref{eq:det2}. We now choose $\lambda = s \Vert \rho \Vert_{L^1} / (2 \Vert \rho \Vert_\infty)$ (which minimises the right hand side of Ineq.{} \eqref{eq:det2}) and obtain Ineq.{} \eqref{eq:det1}.
\end{proof}
The last lemma can be used to obtain bounds on averages of resolvents.
\begin{lemma}\label{lemma:averagenorm}
 Let $n \in \mathbb{N}$, $A \in \mathbb{C}^{n \times n}$ an arbitrary matrix, $V \in \mathbb{C}^{n \times n}$ an invertible matrix and $s \in (0,1)$. Let further $0 \leq \rho \in L^1(\RR) \cap L^\infty (\RR)$ with $\supp \rho \subset [-R,R]$ for some $R>0$. Then we have the bounds
\begin{equation} \label{eq:norm_estimate}
\Vert V^{-1} \Vert \leq \frac{\Vert V \Vert^{n-1}}{\abs{\det V}}
\end{equation}
and
\begin{equation} \label{eq:average_norm}
\int_{-R}^R \bigl\Vert (A+rV)^{-1} \bigr\Vert^{s/n} \rho(r) {\rm d}r \leq \frac{\norm{\rho}_{L^1}^{1-s} \norm{\rho}_\infty^s (\Vert A \Vert + R \Vert V \Vert)^{s(n-1)/n}}{s^s 2^{-s} (1-s) \abs{\det V}^{s/n}} .
\end{equation}
\end{lemma}
\begin{proof}
To prove Ineq. \eqref{eq:norm_estimate} let $0< s_1 \leq s_2 \leq \ldots \leq s_n$ be the singular values of $V$. Then we have $\prod_{i=1}^n s_i \leq s_1 s_n^{n-1}$, that is,
\begin{equation}\label{eq:2}
 \frac{1}{s_1} \leq \frac{s_n^{n-1}}{\prod_{i=1}^n s_i} .
\end{equation}
For the norm we have $\Vert V^{-1} \Vert = 1/s_1$ and $\Vert V \Vert = s_n$. For the determinant of $V$ there holds $\abs{\det V} = \prod_{i=1}^n s_i$. Hence, Ineq. \eqref{eq:norm_estimate} follows from Ineq. \eqref{eq:2}. To prove Ineq. \eqref{eq:average_norm} recall that, since $V$ is invertible, the set $\{r \in \RR \colon \text{$A+rV$ is singular}\}$ is a discrete set. Thus, for almost all $r \in [-R,R]$ we may apply Ineq. \eqref{eq:norm_estimate} to the matrix $A+rV$ and obtain
\[
 \bigl\Vert (A+rV)^{-1} \bigr\Vert^{s/n} \leq \frac{(\Vert A \Vert + R \Vert V \Vert)^{s(n-1)/n}}{\abs{\det (A+rV)}^{s/n}} .
\]
Inequality \eqref{eq:average_norm} now follows from Lemma \ref{lemma:det}.
\end{proof}
The assumption that the single-site potential $u$ is monotone at the boundary allows us to use monotone spectral averaging at some stage. For this purpose we cite a special case of \cite[Proposition 3.1]{AizenmanENSS2006}. Recall, a densely defined operator $T$ on some Hilbert space $\mathcal{H}$ with inner product $\langle \cdot , \cdot \rangle_{\mathcal{H}}$ is called \emph{dissipative} if $\im \langle x,Tx \rangle_{\mathcal{H}} \geq 0$ for all $x \in D(T)$.
\begin{lemma} \label{lemma:monotone}
Let $A \in \CC^{n \times n}$ be a dissipative matrix, $V \in \RR^{n \times n}$ diagonal and strictly positive definite and $M_1 , M_2 \in \CC^{n \times n}$ be arbitrary matrices. Then there exists a constant $C_{\rm W}$ (independent of $A$, $V$, $M_1$ and $M_2$), such that
\[
\cL \bigl\{ r \in \RR : \lVert M_1 (A + r V)^{-1} M_2 \rVert_{\rm HS} > t \bigr\} \leq C_{\rm W} \lVert M_1 V^{-1/2} \rVert_{\rm HS} \lVert M_2 V^{-1/2} \rVert_{\rm HS} \frac{1}{t} .
\]
Here, $\cL$ denotes the Lebesgue-measure and $\lVert \cdot \rVert_{\rm HS}$ the Hilbert Schmidt norm.
\end{lemma}
As a corollary we have
\begin{lemma} \label{lemma:monotone2}
Let $A \in \CC^{n \times n}$ be a dissipative matrix, $V \in \RR^{n \times n}$ diagonal and strictly positive definite and $M_1 , M_2 \in \CC^{n \times n}$ be arbitrary matrices. Then there exists a constant $C_{\rm W}$ (independent of $A$, $V$, $M_1$ and $M_2$), such that
\[
\int_\RR \lVert M_1 (A + r V)^{-1} M_2 \rVert^s \rho (r) \drm r \leq
\frac{(n C_{\rm W} \lVert M_1 V^{-1/2} \rVert \lVert M_2 V^{-1/2} \rVert \lVert \rho \rVert_\infty )^s}{1-s} .
\]
\end{lemma}
\begin{proof}
 First note that for a matrix $T \in \CC^{n \times n}$ we have $\lVert T \rVert \leq \lVert T \rVert_{\rm HS} \leq \sqrt{n} \lVert T \rVert$. With the use of the layer cake representation, see e.\,g. \cite[p. 26]{LiebL2001}, and Lemma~\ref{lemma:monotone} we obtain for all $\kappa > 0$
\begin{align*}
I &= \int_\RR \lVert M_1 (A + r V)^{-1} M_2 \rVert^s \rho (r) \drm r
= \int_0^\infty \int_\RR \mathbf{1}_{\{\lVert M_1 (A + r V)^{-1} M_2 \rVert^s > t \}} \rho (r) \drm r \drm t \\[1ex]
& \leq \kappa + \int_\kappa^\infty \lVert \rho \rVert_\infty n C_{\rm W} \lVert M_1 V^{-1/2} \rVert \lVert M_2 V^{-1/2} \rVert \frac{1}{t^{1/s}} \drm t \\[1ex]
& = \kappa + \lVert \rho \rVert_\infty n C_{\rm W} \lVert M_1 V^{-1/2} \rVert \lVert M_2 V^{-1/2} \rVert \frac{s}{1-s} \kappa^{(s-1)/s} .
\end{align*}
If we choose $\kappa = (\lVert \rho \rVert_\infty n C_{\rm W} \lVert M_1 V^{-1/2} \rVert \lVert M_2 V^{-1/2} \rVert)^s$ we obtain the statement of the lemma.
\end{proof}
Note that all lemmata so far concerned finite dimensional matrices only.
In order to use them for our infinite dimensional operator $G_\omega (z)$ we will apply a special case of the Schur complement formula (also known as Feshbach formula or Grushin problem), see e.\,g. \cite[appendix]{BellissardHS2007}.
\begin{lemma} \label{lemma:schur1}
Let $\Lambda \subset \Gamma \subset \ZZ^d$ and $\Lambda$ finite. Then we have for all $z \in \CC\setminus\RR$ the identity
\begin{equation*}
 \Pro_{\Lambda}^\Gamma(H_\Gamma - z)^{-1} (\Pro_{\Lambda}^\Gamma)^* = \bigl(H_{\Lambda} - B_\Gamma^\Lambda - z \bigr)^{-1} ,
\end{equation*}
where $B_\Gamma^\Lambda : \ell^2 (\Lambda) \to (\Lambda)$ is specified in Eq. \eqref{eq:bij}. Moreover, the operator $B_\Gamma^\Lambda$ is independent of $V_\omega (k)$, $k \in \Lambda$.
\end{lemma}
\begin{proof}
An application of the Schur complement formula gives
\begin{equation*} 
 \Pro_\Lambda^\Gamma(H_\Gamma - z)^{-1} (\Pro_\Lambda^\Gamma)^* =
 \Bigl[ H_\Lambda-z - \Pro_\Lambda^\Gamma \Delta_\Gamma (\Pro_{\Gamma\setminus\Lambda}^\Gamma)^* \bigl(H_{\Gamma \setminus \Lambda} - z \bigr)^{-1} \Pro_{\Gamma\setminus\Lambda}^\Gamma \Delta_\Gamma (\Pro_{\Lambda}^\Gamma)^* \Bigr]^{-1} ,
\end{equation*}
compare, e.\,g., \cite[Appendix]{BellissardHS2007}. For $\Lambda \subset \Gamma \subset \ZZ^d$ we define
\begin{subequations} \label{eq:bij}
\begin{equation} \label{eq:bij1}
 B_\Gamma^\Lambda := \Pro_\Lambda^\Gamma \Delta_\Gamma (\Pro_{\Gamma\setminus\Lambda}^\Gamma)^* \bigl(H_{\Gamma \setminus \Lambda} - z \bigr)^{-1} \Pro_{\Gamma\setminus\Lambda}^\Gamma \Delta_\Gamma (\Pro_{\Lambda}^\Gamma)^* .
\end{equation}
For the matrix elements of $B_\Gamma^\Lambda$ one calculates
\begin{equation} \label{eq:bij2}
\sprod{\delta_x^{\Lambda}}{B_\Gamma^\Lambda \delta_y^{\Lambda}} =
\begin{cases}
   \quad \!\!0  & \text{if $x \not \in \partial^{\rm i} \Lambda \vee y \not \in \partial^{\rm i} \Lambda$,} \\[1ex]
   \sum\limits_{\genfrac{}{}{0pt}{2}{k \in \Gamma\setminus\Lambda :}{\abs{k-x} = 1} }
   \sum\limits_{\genfrac{}{}{0pt}{2}{l \in \Gamma\setminus\Lambda :}{\abs{l-y} = 1} } G_{\Gamma\setminus\Lambda}(z;k,l) & \text{if $x \in \partial^{\rm i} \Lambda \wedge y \in \partial^{\rm i} \Lambda$.}
\end{cases}
\end{equation}
\end{subequations}
$G_{\Gamma\setminus\Lambda}$ is independent of $V_\omega (k)$, $k\in \Lambda$. Thus it is $B_\Gamma^\Lambda$ likewise.
\end{proof}
\begin{lemma}\label{lemma:schur2}
 Let $\Gamma \subset \ZZ^d$ and $\Lambda_1 \subset \Lambda_2 \subset \Gamma$. We assume that $\Lambda_1$ and $\Lambda_2$ are finite sets and that $\partial^{\rm i} \Lambda_2 \cap \Lambda_1 = \emptyset$. Then we have for all $z \in \CC\setminus\RR$ the identity
\begin{multline*}
\Pro_{\Lambda_1}^\Gamma(H_\Gamma - z)^{-1} \bigl(\Pro_{\Lambda_1}^\Gamma\bigr)^*
= \Bigl[ H_{\Lambda_1} - z \\ - \Pro_{\Lambda_1} \Delta \Pro_{\Lambda_2 \setminus \Lambda_1}^*
\Bigl(H_{\Lambda_2 \setminus \Lambda_1} - z - \Pro_{\Lambda_2 \setminus \Lambda_1}^{\Lambda_2} B_\Gamma^{\Lambda_2} \bigl(\Pro_{\Lambda_2 \setminus \Lambda_1}^{\Lambda_2}\bigr)^*\Bigr)^{-1}
\Pro_{\Lambda_2 \setminus \Lambda_1} \Delta \Pro_{\Lambda_1}^*
\Bigr]^{-1}.
\end{multline*}
\end{lemma}
\begin{proof}
We decompose $\Lambda_2 = \Lambda_1 \cup (\Lambda_2 \setminus \Lambda_1)$ and notice that $\langle \delta_x , B_\Gamma^{\Lambda_2} \delta_y \rangle = 0$ if $x \in \Lambda_1$ or $y \in \Lambda_1$ by Eq. \eqref{eq:bij2}. Due to this decomposition we write $H_{\Lambda_2} - z - B_{\Gamma}^{\Lambda_2}$ as the block operator matrix
\[
H_{\Lambda_2} - z - B_{\Gamma}^{\Lambda_2} =
\begin{pmatrix}
  H_{\Lambda_1} - z & -\Pro_{\Lambda_1} \Delta \Pro_{\Lambda_2 \setminus \Lambda_1}^* \\[2ex]
  -\Pro_{\Lambda_2 \setminus \Lambda_1} \Delta \Pro_{\Lambda_1}^* & H_{\Lambda_2 \setminus \Lambda_1} - z - \Pro_{\Lambda_2 \setminus \Lambda_1}^{\Lambda_2} B_\Gamma^{\Lambda_2} \bigl(\Pro_{\Lambda_2 \setminus \Lambda_1}^{\Lambda_2}\bigr)^*
\end{pmatrix} .
\]
The Schur complement formula gives $\Pro_{\Lambda_1}^{\Lambda_2}(H_{\Lambda_2} - z - B_{\Gamma}^{\Lambda_2})^{-1} (\Pro_{\Lambda_1}^{\Lambda_2})^* = S^{-1}$ where $S$ equals
\[
 H_{\Lambda_1} - z - \Pro_{\Lambda_1} \Delta \Pro_{\Lambda_2 \setminus \Lambda_1}^*
\bigl(H_{\Lambda_2 \setminus \Lambda_1} - z - \Pro_{\Lambda_2 \setminus \Lambda_1}^{\Lambda_2} B_\Gamma^{\Lambda_2} \bigl(\Pro_{\Lambda_2 \setminus \Lambda_1}^{\Lambda_2}\bigr)^*\bigr)^{-1}
\Pro_{\Lambda_2 \setminus \Lambda_1} \Delta \Pro_{\Lambda_1}^* .
\]
Since $\Pro_{\Lambda_1}^{\Lambda_2}(H_{\Lambda_2} - z - B_{\Gamma}^{\Lambda_2})^{-1} (\Pro_{\Lambda_1}^{\Lambda_2})^* = \Pro_{\Lambda_1}^\Gamma (H_\Gamma - z)^{-1} (\Pro_{\Lambda_1}^\Gamma)^*$ by Lemma~\ref{lemma:schur1}, we obtain the statement of the lemma.
\end{proof}
%
%
%
%
%
%
%
%
\section{Exponential decay of fractional moments through the finite volume
criterion} \label{sec:exp_decay}
In this section we show that the so called finite volume criterion implies exponential decay of the Green function. Together with the a-priori bound from Lemma \ref{lemma:bounded} this gives us Theorem~\ref{thm:result1}, which will be proven at the end of this section.
\par
We shall consider ``depleted'' Hamiltonians to formulate a geometric
resolvent formula. Such Hamiltonians are obtained by setting to zero
the ``hopping terms'' of the Laplacian along a collection of bonds.
More precisely, let $\Lambda \subset \Gamma \subset \ZZ^d$ be
arbitrary sets. We define the depleted Laplace operator
$\Delta_\Gamma^\Lambda :\ell^2 (\Gamma) \to \ell^2 (\Gamma)$ by
\begin{equation*} \label{eq:de1}
 \sprod{\delta_x}{\Delta_\Gamma^\Lambda \delta_y} :=
\begin{cases}
  0 & \text{if $x \in \Lambda$, $y \in \Gamma \setminus \Lambda$ or $y \in \Lambda$,
  $x \in \Gamma \setminus \Lambda$} , \\
  \bigl \langle \delta_x , \Delta_\Gamma \delta_y \bigr \rangle & \text{else} .
\end{cases}
\end{equation*}
In other words, the hopping terms which connect $\Lambda$ with
$\Gamma \setminus \Lambda$ or vice versa are deleted. The depleted
Hamiltonian $H_\Gamma^\Lambda : \ell^2 (\Gamma) \to \ell^2 (\Gamma)$
is then defined by
\begin{equation*} \label{eq:depl}
 H_\Gamma^\Lambda := -\Delta_\Gamma^\Lambda + V_\Gamma .
\end{equation*}
Let further $T_\Gamma^\Lambda := \Delta_\Gamma -
\Delta_\Gamma^\Lambda$ be the difference between the the ``full''
Laplace operator and the depleted Laplace operator. For $z \in \CC \setminus \RR$ and $x,y \in \Gamma$ we use the notation $G_\Gamma^\Lambda (z) = (H_\Gamma^\Lambda - z)^{-1}$ and $G_\Gamma^\Lambda (z;x,y) = \bigl \langle \delta_x, G_\Gamma^\Lambda(z) \delta_y \bigr \rangle$. To formulate a geometric resolvent formula we apply the second resolvent identity and obtain for arbitrary sets $\Lambda \subset \Gamma \subset \ZZ^d$
\begin{equation}
G_\Gamma (z) = G_\Gamma^\Lambda (z) + G_\Gamma (z) T_\Gamma^\Lambda  G_\Gamma^\Lambda (z) = G_\Gamma^\Lambda (z) + G_\Gamma^\Lambda (z) T_\Gamma^\Lambda G_\Gamma (z) .
\label{eq:firstorder}
\end{equation}
%
For our purposes it will be necessary to use an iterated version of this formula.
Namely, the two applications of the resolvent identity give
\begin{equation}
G_{\Gamma} (z)  = G_{\Gamma}^{\Lambda} (z) + G_{\Gamma }^{\Lambda} (z) T^{\Lambda} G_{\Gamma}^{\Lambda} (z) + G_{\Gamma }^{\Lambda} (z)  T^{\Lambda} G_{\Gamma} (z) T^{\Lambda} G_{\Gamma}^{\Lambda} (z) .
\label{eq:secondorder}
\end{equation}
\begin{remark} \label{remark:depleted}
Notice that $G_\Gamma^\Lambda (z;x,y) = G_\Lambda (z;x,y)$ if $x,y \in \Lambda$, $G_\Gamma^\Lambda (z;x,y) = 0$ if $x \in \Lambda$ and $y \not \in \Lambda$ or vice versa, and that $G_\Gamma^\Lambda (z) = G_\Gamma^{\Lambda^{\rm c}} (z)$. If $\Gamma \setminus \Lambda$ decomposes into at least two components which are not connected, and $x$ and $y$ are not in the same component, then we also have $G_\Gamma^\Lambda (z;x,y) = 0$.
\par
Since $\Gamma$ is not necessarily the whole lattice
$\ZZ^d$, it may be that terms of the type $G_{\Gamma} (z;i,j)$ occur for
some $\Gamma \subset \ZZ^d$ and some $i \not \in \Gamma$ or $j \not \in \Gamma$. In this case we use the
convention that $G_\Gamma (z;i,j) = 0$.
\end{remark}
To formulate the results of this section we will need the following
notation: For finite $\Gamma \subset \ZZ^d$ we denote by $\diam \Gamma$ the diameter of $\Gamma$ with respect to the supremum norm, i.\,e. $\diam \Gamma = \sup_{x,y\in \Gamma} \lvert x-y \rvert_\infty$. Let $\Gamma \subset \ZZ^d$, fix $L \ge \diam \Theta + 2$, let $\Lambda_L = [-L,L]^d \cap \ZZ^d$ be a cube of size $2L+1$, let
\[
B = \partial^{\rm i} \Lambda_L,
\]
and define the sets
\[
\hat \Lambda_x = \{ k \in \Gamma : k \in \Theta_b \text{ for some $b \in \Lambda_{L,x}$} \}
\]
and
\begin{equation} \label{eq:Wx}
 \hat W_x = \{ k \in \Gamma : k \in \Theta_b \text{ for some $b \in B_x$} \} .
\end{equation}
Recall that for $\Gamma \subset \ZZ^d$ we denote by $\Gamma_x = \Gamma + x = \{k \in \ZZ^d : k-x \in \Gamma\}$ the translate of $\Gamma$.
Hence $(\Lambda_L)_x=\Lambda_{L,x}$ and
$\hat W_x$ is the union of translates of $\Theta$ along the sides of $B_x$, restricted to the set $\Gamma$. For $\Gamma \subset \ZZ^d$ we can now introduce the sets
\[
\Lambda_x: = \hat \Lambda_x^+ \cap \Gamma \quad \text{and} \quad W_x = \hat W_x^+ \cap \Gamma
\]
which will play a role in the assertions below.
\begin{theorem}[Finite volume criterion] \label{thm:exp_decay}
Suppose that Assumption \ref{ass:monotone} is satisfied, let $\Gamma \subset \ZZ^d$, $z
\in \CC \setminus \RR$ with $\lvert z \rvert \leq m$ and $s \in (0,1/3)$. Then there
exists a constant $B_s$ which depends only on $d$, $\rho$, $u$, $m$, $s$,
such that if the condition
\begin{equation}\label{eq:fin_cond}
b_s(\lambda, L,\Lambda): = \frac{B_s L^{3(d-1)} \Xi_s (\lambda)}{\lambda^{2s/(2\lvert
\Theta \rvert)}}\, \sum_{w\in\partial^{\rm o} W_x}\mathbb{E}
\bigl\{\lvert G_{\Lambda\setminus W_x} (z;x,w)\rvert^{s/(2\lvert
\Theta \rvert)}\bigr\}< b
\end{equation}
is satisfied for some $b \in (0,1)$, arbitrary $\Lambda \subset \Gamma$, and all $x\in
\Lambda$, then for all $x,y \in \Gamma$
\begin{equation*} 
\mathbb{E} \bigl\{\lvert G_\Gamma (z;x,y)\rvert^{s/(2\lvert \Theta
\rvert)} \bigr\}\leq A \euler^{-\mu|x-y|_\infty} .
\end{equation*}
Here
\[
A=\frac{C_s \Xi_s (\lambda)}{b} \quad \text{and} \quad
\mu=\frac{\lvert \ln b \rvert}{L+\diam \Theta + 2}\,,\] with $C_s$ inherited
from the a-priori bound (Lemma \ref{lemma:bounded}).
\end{theorem}
\begin{remark} \label{remark:finite_volume}
Note that $\Gamma \setminus W_x$ decomposes into  two components
which are not connected, so that the sum in \eqref{eq:fin_cond} runs
over the sites $r$ related to only one of these components, which is
always compact, regardless of the choice of $\Gamma$. It then
follows that in order to establish the exponential falloff of the
Green function it suffices to consider the decay properties of the
Green function for the Hamiltonians defined on finite sets. The
finite volume criterion derives its name from this fact.
\end{remark}
%
%
%
The strategy for the proof is reminiscent of the one developed in
\cite{AizenmanFSH2001} and is aimed to derive a following bound on
the average Green function.
\begin{lemma} \label{lemma:iteration1}
Let $\Gamma \subset \ZZ^d$, $s \in (0, 1/3)$, $m > 0$, Assumption \ref{ass:monotone} be satisfied and $b_s (\lambda,L,\Lambda)$ be the constant from Theorem \ref{thm:exp_decay}. Then we have for all $x,y \in \Gamma$ with $y \not \in \Lambda_x$ and all $z \in \CC \setminus \RR$ with $\lvert z \rvert \leq m$ the bound
\begin{equation}\label{eq:protobound}
\mathbb{E} \bigl\{\lvert G_\Gamma (z;x,y)\rvert^{\frac{s}{2\lvert
\Theta \rvert}}\bigr\}\leq
\frac{b_s(\lambda, L,\Gamma)}{|\partial^{\rm o} \Lambda_x|} \sum_{r\in\partial^{\rm o}
\Lambda_x} \mathbb{E}\bigl\{ \lvert G_{\Gamma\setminus
\Lambda_x}(z;r,y)\rvert^{\frac{s}{2\lvert \Theta \rvert}} \bigr\} .
\end{equation}
\end{lemma}
\begin{remark}
Equation~\eqref{eq:protobound} with $b_s(\lambda, L,\Gamma) < b < 1$ is akin to
the statement that the expectation $\EE \{ |G_\Gamma(z;x,\cdot)|^s
\}$ is a strictly subharmonic function, and thus, since it is also
uniformly bounded by the a-priori bound from
Lemma~\ref{lemma:bounded}, it decays exponentially. Indeed, since
the sum is normalized by the prefactor $1/\lvert \partial^{\rm o}
\Lambda_x \rvert$, Ineq.~\eqref{eq:protobound} permits to improve
the a-priori bound by the factor $b_s(\lambda, L,\Gamma) <b$. Furthermore,
the inequality may be iterated, each iteration resulting in an
additional factor of $b_s(\lambda, L,\Gamma)$. Also note that each
iteration step brings in Green functions that correspond to modified
domains.
\end{remark}
The finite volume criterion is a direct corollary of Lemma~\ref{lemma:iteration1}:
\begin{proof}[Proof of Theorem \ref{thm:exp_decay}]
Inequality~\eqref{eq:protobound} can be iterated as long as
the resulting sequences ($x, r^{(1)}, \ldots, r^{(n)}$) do not get
closer to $y$ than the distance $\tilde L = L + \diam \Theta + 2$.
\par
If $\lvert x-y \rvert_\infty \geq \tilde L$, we iterate Ineq.~\eqref{eq:protobound}
exactly $\lfloor \lvert x-y \rvert_\infty/ \tilde L \rfloor$ times, use the a-priori bound from
Lemma~\ref{lemma:bounded} and obtain
\begin{equation*}
\EE \Bigl \{ \bigl| G_\Gamma(z;x,y) \bigr|^{\frac{s}{2\lvert \Theta \rvert}} \Bigr\}
\leq
C_s \Xi_s (\lambda) \cdot b^{\textstyle \lfloor \lvert x-y \rvert_\infty/ \tilde L \rfloor}
\leq
\frac{C_s \Xi_s (\lambda)} {b} {\rm e}^{- \mu |x-y|_\infty} , \label{eq:thm1optimalbound}
\end{equation*}
with $\mu = \lvert \ln b \rvert / \tilde L$. If
$\lvert x-y \rvert_\infty < \tilde L$, we use
Lemma~\ref{lemma:bounded} and see that
\[
 \EE \Bigl \{ \bigl| G_\Gamma(z;x,y) \bigr|^{s/(2\lvert \Theta \rvert)} \Bigr\}
  \leq C_s \Xi_s (\lambda) \leq \frac{C_s \Xi_s (\lambda)}{b} {\rm e}^{- \mu |x-y|_\infty} . \qedhere
\]
\end{proof}
To facilitate the proof of Lemma \ref{lemma:iteration1} we introduce
some extra notation first. Namely, for a set $\Lambda \subset
\ZZ^d$, we define the bond-boundary $\partial^{\rm B} \Lambda$ of
$\Lambda$ as
\[
\partial^{\rm B} \Lambda = \left\{(u,u') \in \ZZ^d \times \ZZ^d :
u\in \Lambda,\ u'\in \ZZ^d \setminus \Lambda,\ \text{and} \ \lvert
u-u' \rvert_1 = 1\right\}\,.
\]
\begin{proof}[Proof of Lemma~\ref{lemma:iteration1}]
Fix $x,y \in \Gamma$ with $y \not \in \Lambda_x$ and set $n = 2\lvert \Theta \rvert$.
It follows from our definition, that the randomness of $H_\Gamma$ at sites $\partial^{\rm o} \hat W_x \cap \Gamma$ does not depend on the random variables $\omega_b$ for any $b\in B_x$, and depends {\it monotonically} on the random variables $\omega_k$ for $k\in \partial^{\rm o} B_x$ (by Assumption \ref{ass:monotone}). A similar statement holds for the randomness at sites $\partial^{\rm o} W_x \cap \Gamma$. We also note that $x,y \not \in W_x$ by our definition of $L$ and since $0 \in \Theta$.
We now choose $\Lambda = \hat W_x$ in Eq.~\eqref{eq:secondorder} and compute the Green function at $(x,y)$:
\begin{multline*}
G_{\Gamma}(z;x,y) =  G_{\Gamma}^{\hat W_x}(z;x,y) + \langle \delta_x , G_{\Gamma}^{\hat W_x} (z) T_{\Gamma}^{\hat W_x} G_{\Gamma}^{\hat W_x} (z) \delta_y \rangle \\
+ \langle \delta_x , G_{\Gamma}^{\hat W_x} (z) T_{\Gamma}^{\hat W_x} G_{\Gamma} (z) T_{\Gamma}^{\hat W_x} G_{\Gamma}^{\hat W_x} \delta_y \rangle .
\label{eq:secondordermatelts}
\end{multline*}
Using Remark \ref{remark:depleted} one can easily check that the first two contributions vanish, thus
\begin{equation}
G_{\Gamma}(z;x,y) = \sum_{\genfrac{}{}{0pt}{2}{(u',u) \in \partial^{\rm B} \hat W_x}{(v,v') \in  \partial^{\rm B} \hat W_x}}
    G_{\Gamma}^{\hat W_x}(z;x,u) G_\Gamma(z;u',v) G_{\Gamma}^{\hat W_x}(z;v',y) .
\label{eq:diagramatic}
\end{equation}
\begin{figure}[t]
\begin{tikzpicture}[scale=0.3]
\usetikzlibrary{patterns}
\draw[very thick] (10,10) rectangle (-10,-10);
\draw[very thin] (0,12)--(12,12)--(12,-4)--(14,-4)--(14,14)--(0,14)--(0,16)--(16,16)--(16,-6)--(12,-6)--(12,-12)--(0,-12)--(0,-8)--(8,-8)--(8,-6)--(0,-6)--(0,-4)--(8,-4)--(8,8)--(0,8)--(0,12);
\filldraw[very thin, opacity=0.07] (0,12)--(12,12)--(12,-4)--(14,-4)--(14,14)--(0,14)--(0,16)--(16,16)--(16,-6)--(12,-6)--(12,-12)--(0,-12)--(0,-8)--(8,-8)--(8,-6)--(0,-6)--(0,-4)--(8,-4)--(8,8)--(0,8)--(0,12);
\draw[dotted] (0,-10)--(0,15);
\filldraw (0,0) circle (5pt); \draw (1,0) node {$x$};
\filldraw (5,7.3) circle (5pt); \draw (3.9,7.3) node {$u$};
\filldraw (5,8.7) circle (5pt); \draw (4,8.7) node {$u'$};
\filldraw (15.3,4) circle (5pt); \draw (15.3,2.85) node {$v$};
\filldraw (16.7,4) circle (5pt); \draw (16.7,3) node {$v'$};
\filldraw (22,8) circle (5pt); \draw (23,8) node {$y$};
\draw (0,0)--(5,7.3);
\draw[very thin, opacity=0.5] (5,7.3)--(5,8.7);
\draw[very thick] (5,8.7)--(15.3,4);
\draw[very thin, opacity=0.5] (15.3,4)--(16.7,4);
\draw (16.7,4)--(22,8);
\end{tikzpicture}
\caption{Illustration of the geometric setting and Eq. \eqref{eq:diagramatic}
in the case $d = 2$, $\Gamma = \{x \in \mathbb{Z}^2 : x_1 \geq 0\}$, $x = 0$
and $\Theta = ([-2,2]^2 \cup [4,6]^2)\cap \mathbb{Z}^2$.
The light grey region is the set $\hat W_x$ and the black square is the sphere $B_x$.}
\label{fig:geometric2}
\end{figure}
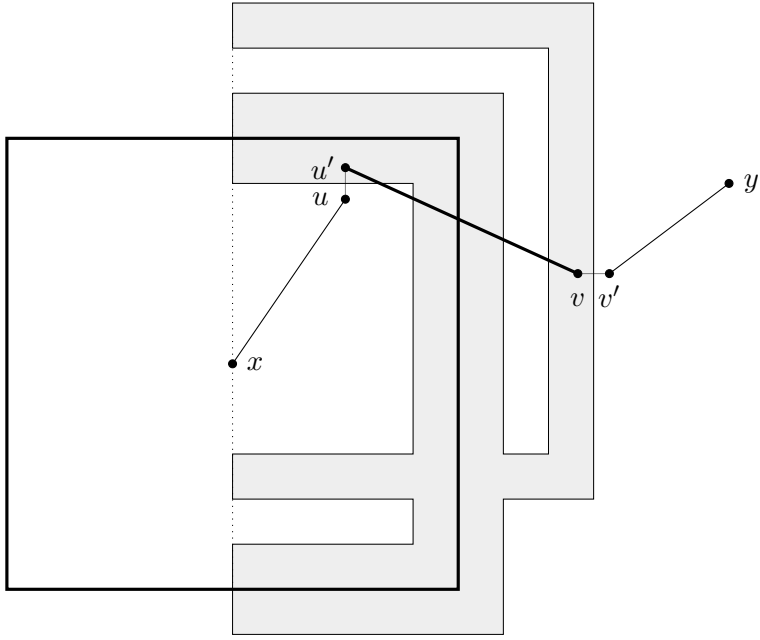

See Fig.~\ref{fig:geometric2} for the geometric setting and an illustration of Eq. \eqref{eq:diagramatic}.
Note that $u, v' \in \partial^{\rm o} \hat W_x$, while $u', v \in
\hat W_x$. By  construction, the set $\Gamma \setminus \hat W_x$
decomposes into at least two components which are not connected:
One of them contains $x$, another $y$. More than two components may occur if $\Gamma$ or $\Theta$ are not connected, see again Fig. \ref{fig:geometric2}. By Remark
\ref{remark:depleted}, the summands in Eq. \eqref{eq:diagramatic}
are only non-zero if $u$ is in the $x$-component of $\Gamma
\setminus \hat W_x$ and $v'$ is in the $y$-component of $\Gamma
\setminus \hat W_x$. This leads us to the definition of a subset of
$\partial^{\rm B} \hat W_x$. We say that $(u,u') \in
\partial_x^{\rm B} \hat W_x$ if $(u,u') \in \partial^{\rm B} \hat
W_x$ and $u'$ is in the $x$-component of $\Gamma \setminus \hat
W_x$. For $\partial_y^{\rm B} \hat W_x$, $\partial_x^{\rm B} W_x$
and $\partial_y^{\rm B} W_x$ we use the analogous definitions.
\par
To get the estimate \eqref{eq:protobound} we want to first average
the fractional moment of the Green function with respect to random
variables $\{\omega_k\}_{k\in B_x^+}$. Note that Lemma~\ref{lemma:bounded} part (a) then guarantees that
\begin{equation} \label{eq:apriori}
\mathbb{E}_{B_x^+} \bigl\{\lvert G_\Gamma (z;u',v)\rvert^{s/n}\bigr\} \leq C_s \Xi_s (\lambda) .
\end{equation}
However, although the first and the last Green functions in \eqref{eq:diagramatic} do not depend on the random variables $\{\omega_k\}_{k\in B_x}$, they still depend on the random variables $\{\omega_k\}_{k\in B_x^+}$. To factor out this dependence, we apply \eqref{eq:firstorder} again, this time with $\Lambda = W_x$. Then we have for $u$, $v'$ as
above the equalities
\begin{align*}
G_{\Gamma}^{\hat W_x} (z;x,u) &= \sum_{(w',w) \in \partial_x^{\rm B} W_x}  G_{\Gamma}^{W_x}(z;x,w)G_{\Gamma}^{\hat W_x}(z;w',u)
\intertext{and}
G_{\Gamma}^{\hat W_x}(z;v',y)  &= \sum_{(r,r') \in \partial_y^{\rm B} W_x} G_{\Gamma}^{\hat W_x}(z;v',r) G_{\Gamma}^{W_x}(z;r',y) .
\end{align*}
Notice that for $w$ and $r'$ as above, the Green functions $G_{\Gamma}^{W_x}(z;x,w)$ and $G_{\Gamma}^{W_x}(z;r',y)$ are independent of $\{\omega_k\}_{k \in B_x^+}$. Putting everything together, we obtain
\begin{multline}
\mathbb{E}_{B_x^+} \bigl\{\lvert G_\Gamma (z;x,y)\rvert^{s/n}\bigr\} \leq
\sum \lvert G_{\Gamma}^{W_x}(z;x,w)\rvert^{s/n} \, \lvert G_{\Gamma}^{W_x}(z;r',y)\rvert^{s/n} \\
\times \mathbb{E}_{B_x^+} \bigl\{\lvert G_{\Gamma}^{\hat W_x}(z;w',u) G_\Gamma (z;u',v) G_{\Gamma}^{\hat W_x}(z;v',r) \rvert^{s/n}\bigr\} ,
\label{eq:optimalbound}
\end{multline}
where  the sum on the right hand side runs over the bonds
\[
(u',u) \in \partial_x^{\rm B} \hat W_x, \ (v,v') \in  \partial_y^{\rm B} \hat W_x, \
(r,r') \in \partial_y^{\rm B} W_x,\ (w',w) \in \partial_x^{\rm B} W_x .
\]
To estimate the expectation of the product on the right hand side we
note that by H\"older inequality it suffices to show that each of
the Green functions raised to the fractional power $3s/n$ and averaged with respect to $B_x^+$ is bounded in an appropriate way. For $\EE_{B_x^+} (|G_\Gamma (z;u',v)|^{3s/n})$ this follows from the a-priory bound \eqref{eq:apriori}.
For the remaining two Green functions it seems at the first glance that we have a problem,
since only we average over $\{\omega_k\}_{k\in B_x^+}$, and Lemma~\ref{lemma:bounded}
in this context requires averaging with respect to $\{\omega_k\}_{k\in
B_x^{++}}$. What comes to our rescue is Assumption \ref{ass:monotone},
which ensures that the dependence on $\{\omega_k\}_{k\in B_x^+}$ is
actually monotone for these Green functions, and the standard
argument of \cite{AizenmanENSS2006} for the monotone case establishes
the required bounds. More precisely, we argue as follows. Since $w',u \in \Gamma \setminus \hat W_x$, we have due to Remark~\ref{remark:depleted} that
\[
G_{\Gamma}^{\hat W_x}(z;w',u) = G_{\Gamma\setminus\hat W_x}(z;w',u).
\]
Notice that $w',u \in \partial^{\rm o} \hat W_x$. Hence there are $b_{1},b_{2} \in \partial^{\rm o} B_x$, such that $w' \in \Theta_{b_1} \cap (\Gamma\setminus\hat W_x) \subset \partial^{\rm i} \Theta_{b_1}$ and $u \in \Theta_{b_2} \cap (\Gamma\setminus\hat W_x) \subset \partial^{\rm i} \Theta_{b_2}$. For the Green function at $(v',r)$ there exist $b_3,b_4 \in \partial^{\rm o} B_x$ with analoguous properties. Thus we may apply Lemma~\ref{lemma:bounded} part (b) and obtain for all $t \in (0,1)$
\begin{equation*}
 \mathbb{E}_{B_x^+} \bigl\{\lvert G_{\Gamma}^{\hat W_x}(z;w',u)
 \rvert^{t}\bigr\} \leq  D_t \lambda^{-t} \quad \text{and} \quad \mathbb{E}_{B_x^+}
 \bigl\{\lvert G_{\Gamma}^{\hat W_x}(z;v',r)\rvert^{t}\bigr\} \leq D_t\lambda^{-t} .\label{eq:mon_est}
\end{equation*}
The net result is a bound
\[
 \mathbb{E}_{B_x^+} \bigl\{\lvert G_{\Gamma}^{\hat W_x}(z;w',u) G_\Gamma (z;u',v) G_{\Gamma}^{\hat W_x}(z;v',r) \rvert^{s/n}\bigr\} \leq E_s \lambda^{-\frac{2s}{n}} \Xi_s (\lambda)
\]
where $E_s = \max\{D_{3s/n} , C_{3s}\}$. Substitution into Ineq. \eqref{eq:optimalbound} leads to the estimate
\begin{multline} \label{eq:optimalbound1}
 \mathbb{E}_{B_x^+} \bigl\{\lvert G_\Gamma (z;x,y)\rvert^{s/n}\bigr\} \leq E_s \lambda^{-\frac{2s}{n}} \Xi_s (\lambda) \lvert \partial_x^{\rm B} \hat W_x \rvert \lvert \partial_y^{\rm B} \hat W_x \rvert \\
\times \sum_{\genfrac{}{}{0pt}{2}{(r,r') \in \partial_y^{\rm B} W_x}{(w',w) \in \partial_x^{\rm B} W_x}} \lvert G_{\Gamma}^{W_x}(z;x,w)\rvert^{s/n} \, \lvert G_{\Gamma}^{W_x}(z;r',y)\rvert^{s/n} .
\end{multline}
Now we are in position to perform the expectation with respect to
the rest of random  variables. Note that the two remaining  Green
functions in \eqref{eq:optimalbound1} are stochastically independent. We take expectation in Ineq. \eqref{eq:optimalbound} and use Remark \ref{remark:depleted} to get
\[
\mathbb{E}\bigl\{\lvert G_\Gamma (z;x,y)\rvert^{s/n}\bigr\}
 \leq  \frac{E_s \tilde \Phi (\Theta, L)}{\lambda^{2s/n} \Xi_s^{-1} (\lambda)}  \cdot
 \sum_{(r,r')\in\partial_y^{\rm B}  W_x}
 \mathbb{E}\bigl
 \{ \lvert G_{\Gamma\setminus W_x}(z;r',y)\rvert^{s/n} \bigr\} \\[1ex]
\]
where
 \[
\tilde \Phi (\Theta, L) = \lvert \partial_x^{\rm B} \hat W_x \rvert
\lvert
\partial_y^{\rm B} \hat W_x \rvert \sum_{ (w',w) \in
\partial_x^{\rm B} W_x} \EE \bigl\{ \lvert G_{\Gamma \setminus W_x}(z;x,w)\rvert^{s/n} \bigr\} .
\]
Now we use the fact that each point of  $\partial^{\rm o} \Lambda_x$ shares the bond with at most $2d$ neighbors. Hence, if we set
\[
\Phi (\Theta, L) = 4d^2\,\lvert \partial_x^{\rm B} \hat W_x \rvert
\lvert
\partial_y^{\rm B} \hat W_x \rvert \lvert\partial^{\rm o} \Lambda_x\rvert\,
\sum_{w \in\partial^{\rm o} W_x}
\mathbb{E}\bigl \{ \lvert G_{\Gamma\setminus W_x}(z;x,w)
\rvert^{s/n} \bigr\} ,
\]
we have the estimate
\begin{equation*}
\mathbb{E}\bigl\{\lvert G_\Omega (z;x,y)\rvert^{s/n}\bigr\}
 \leq  \frac{E_s \Phi (\Theta, L) }{ \lambda^{2s/n} \Xi_s^{-1} (\lambda)}
\frac{1}{|\partial^{\rm o} \Lambda_x|} \sum_{r \in\partial^{\rm o}
\Lambda_x} \mathbb{E}\bigl \{ \lvert G_{\Gamma\setminus
\Lambda_x}(z;r,y) \rvert^{s/n} \bigr\} .
\end{equation*}
Finally, we can bound $\lvert \partial_x^{\rm B} \hat W_x \rvert $,
$\lvert\partial_y^{\rm B} \hat W_x \rvert$ and $\lvert\partial^{\rm o} \Lambda_x\rvert$
by $C_{d,\Theta} L^{d-1}$ with a constant $C_{d,\Theta}$ depending only on $d$ and $\Theta$.
Lemma \ref{lemma:iteration1} now follows by putting everything together.
\end{proof}
\begin{proof}[Proof of Theorem \ref{thm:result1}]
Notice that by Assumption \ref{ass:monotone}
the random potential is uniformly bounded.
Thus $K := \sup_{\omega \in \Omega} \lVert H_\omega \rVert <\infty$.
Choose $M \geq 1 $ and $ m= K+M$. For $|z| \leq m$ and each $b \in (0,1)$
we infer from the a-priori bound (Lemma \ref{lemma:bounded})
that condition \eqref{eq:fin_cond} from Theorem~\ref{thm:exp_decay} is satisfied if
$\lambda$ sufficiently large.

For $|z| \geq m$ we have $\dist(z,\sigma (H_{\Gamma}))\geq M\geq 1$
for all $\omega$. A Combes-Thomas argument (see \cite{CombesT-73},
or Section 11.2 in \cite{Kirsch-08} for an explicit calculation in the discrete setting)
gives the bound
\[
 |G_\Gamma(z; x,y)| \le \frac{2}{M} \euler^{-\gamma |x-y|_1}
\]
for $|z| \geq m$ and arbitrary $x,y \in \Gamma$, where $\gamma :=\min \big(1, \ln \frac{M}{4d}\big)$. Now taking first the fractional power and then the mathematical
expectation gives the desired estimate on
$ \mathbb{E} \bigl\{\lvert G_\Gamma (z;x,y)\rvert^{s/(2\lvert \Theta \rvert)}\bigr\}$.
This finishes the proof.
\end{proof}
%
%
%
%
%
%
%
%
\section{Exponential localization and application to the strong disorder regime} \label{sec:loc}
In this section we prove exponential localization in the case of
sufficiently large disorder, i.\,e. that the continuous spectrum of
$H_\omega$ is empty almost surely and that the eigenfunctions
corresponding to the eigenvalues of $H_\omega$ decay exponentially
at infinity.

The existing proofs of localization via the fractional moment method
use either the Simon Wolff criterion, see e.\,g.
\cite{SimonW1986,AizenmanM1993,AizenmanFSH2001}, or the
RAGE-Theorem, see e.\,g.
\cite{Aizenman1994,Graf1994,AizenmanENSS2006}. Neither dynamical nor
spectral localization can be directly inferred from the behavior of
the Green function using the existent methods for our model. The
reason is that the random variables $V_\omega (x)$, $x \in \ZZ^d$,
are not independent, while the dependence of $H_\omega$ on the
i.\,i.\,d. random variables $\omega_k$, $k \in \ZZ^d$, is not
monotone.
\par
However, for the discrete alloy-type model it is possible to show localization using the multiscale analysis. The two ingredients of the multiscale analysis are the initial length scale estimate and the Wegner estimate, compare assumptions (P1) and (P2) of \cite{DreifusK1989}. The initial length scale estimate is implied by the exponential decay of an averaged fractional power of Green function, i.\,e. Theorem \ref{thm:exp_decay}, using Chebyshev's inequality. A Wegner
estimate for the models considered here was established in \cite{Veselic2010}. Thus a variant of the multiscale analysis of \cite{DreifusK1989} yields pure point spectrum with exponential decaying eigenfunctions for almost all configurations of the randomness. We say a variant, since in our case the potential values are independent only for lattice sites having a minimal distance. It has been implemented in detail in the paper \cite{GerminetK2001} for random Schr\"odinger operators in the continuum, and holds similarly for discrete models. See also \cite{Krueger2010} for a proof of localization via MSA for a class of models including ours.
\par
In \cite{ElgartTV2010} we have established a new variant for concluding exponential localization from bounds on averaged fractional powers of Green function without using the multiscale
analysis. This is done by showing that fractional moment bounds
imply the ``typical output'' of the multiscale analysis, i.\,e. the
hypothesis of Theorem 2.3 in \cite{DreifusK1989}. Then one can
conclude localization using existent methods. However, the
assertions in \cite{ElgartTV2010} are tailored to the
one-dimensional discrete alloy-type model. In this section we
present the multidimensional extension of these results. Although
the arguments are similar to the ones in \cite{ElgartTV2010}, we
will give all the proofs for completeness.
\par
For $L >0$ and $x \in \ZZ^d$ we denote by $\Lambda_{L,x} = \{y \in \ZZ^d : \lvert x-y \rvert_\infty \leq L\}$ the cube of side length $2L+1$. Let further $m > 0$ and $E \in
\RR$. A cube $\Lambda_{L,x}$ is called \emph{$(m,E)$-regular} (for a fixed
potential), if $E \not \in \sigma (H_{\Lambda_{L,x}})$ and
\[
 \sup_{w \in \partial^{\rm i} \Lambda_{L,x}} \lvert G_{\Lambda_{L,x}} (E ; x,w) \rvert
\leq \euler^{-m L} .
\]
Otherwise we say that $\Lambda_{L,x}$ is \emph{$(m , E)$-singular}. The next
Proposition states that certain bounds on averaged fractional moments of Green function imply the hypothesis of Theorem 2.3 in \cite{DreifusK1989} (without applying the induction step of the multiscale analysis).
\begin{proposition} \label{prop:replace-msa}
Let $I \subset \RR$ be a bounded interval and $s \in (0,1)$. Assume the following two statements:
\begin{enumerate}[(i)]
\item There are constants $C,\mu \in (0,\infty)$ and $L_0 \in \NN_0$
such that
\[\EE \bigl\{ \lvert G_{\Lambda_{L,k}} (E;x,y) \rvert^{s}
\bigr\}\ \leq \  C \euler^{-\mu \lvert x-y \rvert_\infty}\]
for all $k \in \ZZ^d$, $L \in \NN$, $x,y \in \Lambda_{L,k}$ with
$\lvert x-y \rvert_\infty \geq L_0$, and all $E \in I$.
\item There is a constant $C' \in (0,\infty)$ such that
\[\EE \bigl\{ \lvert G_{\Lambda_{L,k}} (E+\i \epsilon ;x,x) \rvert^{s}
\bigr\} \leq C'\]
for all $k \in \ZZ^d$, $L \in \NN$, $x \in
\Lambda_{L,k}$, $E \in I$ and all $\epsilon\in (0,1]$ .
\end{enumerate}
Then we have for all $L \geq \max\{ 8\ln (8)/\mu , L_0 , -(8/5\mu)\ln (\lvert I \rvert / 2)\}$ and all $x,y \in \ZZ^d$ with $\lvert x-y\rvert_\infty \geq 2L+\diam \Theta + 1$ that
\begin{multline*}
 \PP \{\forall \, E \in I \text{ either $\Lambda_{L,x}$ or $\Lambda_{L,y}$ is
$(\mu/8,E)$-regular} \}  \\ \geq 1- 8 \lvert \Lambda_{L,x} \rvert(C\lvert I \rvert  + 4C'\lvert \Lambda_{L,x} \rvert / \pi ) \euler^{-\mu sL /8} .
\end{multline*}
\end{proposition}
\begin{proof}
Set $n = \diam \Theta + 1$. Fix $L \in \NN$ with $L \geq \max\{8 \ln (8)/\mu , L_0\}$ and $x,y \in \ZZ^d$ such that $\lvert x-y \rvert_\infty \geq 2L+n$. For $\omega \in \Omega$ and $k \in
\{x,y\}$ we define the sets
\begin{align}
 \Delta_\omega^k &:= \{E \in I : \sup_{w \in \partial^{\rm i} \Lambda_{L,k}} \lvert
G_{\Lambda_{L,k}} (E ; k,w) \rvert > \euler^{-\mu L /8}\}, \nonumber \\
\tilde \Delta_\omega^k &:= \{E \in I : \sup_{w \in \partial^{\rm i} \Lambda_{L,k}}
\lvert G_{\Lambda_{L,k}} (E ; k,w) \rvert > \euler^{-\mu L/4 }\}, \nonumber \\
\text{and} \quad \tilde B_k &:= \{\omega \in \Omega : \mathcal{L} \{\tilde \Delta_\omega^k\} >
 \euler^{-5\mu L /8} \} . \label{eq:deltatilde}
\end{align}
For $\omega \in \tilde B_k$ we have
\begin{align*}
\sum_{w \in \partial^{\rm i} \Lambda_{L,k}}  \int_I \lvert G_{\Lambda_{L,k}} (E ; k,w) \rvert^{s/N} \drm E
& \geq
\int_I \sup_{w \in \partial^{\rm i} \Lambda_{L,k}} \lvert G_{\Lambda_{L,k}} (E ; k,w) \rvert^{s/N} \drm E \\[1ex]
& > \euler^{-5 \mu L / 8} \euler^{-\mu L s/ 4} > \euler^{-7 \mu L / 8}.
\end{align*}
Using $L \geq L_0$ and Hypothesis (i) of the assertion, we obtain
\begin{equation*}
 \PP \{\tilde B_k\} <  \lvert \Lambda_{L,k} \rvert \, \lvert I \rvert C \euler^{-\mu L /8} .
\end{equation*}
For $k \in \{x,y\}$ we denote by $\sigma (H_{\Lambda_{L,k}}) =
\{E_{\omega,k}^i\}_{i=1}^{\lvert \Lambda_{L,k} \rvert}$ the spectrum of
$H_{\Lambda_{L,k}}$. We claim that for $k \in \{x,y\}$,
\begin{equation} \label{eq:claim}
 \omega \in \Omega \setminus \tilde B_k \quad \Rightarrow
 \quad \Delta_\omega^k \subset \bigcup_{i=1}^{\lvert \Lambda_{L,k} \rvert}
 \bigl[E_{\omega,k}^i-\delta , E_{\omega,k}^i + \delta \bigr] =:
 I_{\omega,k}(\delta),
\end{equation}
where $\delta = 2\euler^{-\mu L / 8}$. Indeed, suppose that $E\in \Delta_\omega^k$ and $\dist\big(E,\sigma
(H_{\Lambda_{L,k}})\big)>\delta$. Then there exists $w \in \partial^{\rm i} \Lambda_{L,k}$ such
that $\lvert G_{\Lambda_{L,k}} (E;k,w) \rvert > \euler^{-\mu L / 8}$.
For any $E'$ with $\lvert E-E'\rvert \le 2\euler^{-5\mu L / 8}$ we have
$\delta -\lvert E-E'\rvert\ge \euler^{-\mu L / 8}\ge 2\euler^{-3\mu L / 8}  $ since $L > 8 \ln (8) / \mu$.
Moreover, the first resolvent identity and the estimate $\lVert (H-E)^{-1} \rVert \leq \dist
(E,\sigma (H))^{-1}$ for selfadjoint $H$ and $E \in \CC\setminus \sigma (H)$ implies
\begin{align*}
\lvert G_{\Lambda_{L,k}} (E ; k,w) - G_{\Lambda_{L,k}} (E' ; k,w)\rvert &\leq \ \lvert E-E'\rvert \cdot \lVert G_{\Lambda_{L,k}} (E)\rVert\cdot\lVert G_{\Lambda_{L,k}} (E') \rVert \\[1ex]
& \leq \frac{1}{2} \euler^{-\mu L / 8} ,
\end{align*}
 and hence
\[
 \lvert G_{\Lambda_{L,k}} (E' ; k,w)\rvert
\ > \ \frac{\euler^{-\mu L / 8}}{2} \geq \euler^{-\mu L / 4}
\]
for $L \geq 8 \ln (8) / \mu$.
We infer that $[E-2\euler^{-5\mu L / 8},E+2\euler^{-5\mu L / 8}]\cap I \subset  \tilde \Delta_\omega^k$ and conclude $\mathcal{L} \{\tilde \Delta_\omega^k\} \ge 2\euler^{-5\mu L / 8}$, since $\lvert I \rvert \geq 2 \euler^{-5\mu L / 8}$ by assumption.
This is however impossible if
$\omega\in \Omega \setminus \tilde B_k$ by \eqref{eq:deltatilde},
hence the claim \eqref{eq:claim} follows.

In the following step we use Hypothesis (ii) of the assertion to
deduce a Wegner-type estimate. Let $[a,b] \subset I$ with $0<b-a
\leq 1$. We denote by $P_{[a,b]} (H_{\Lambda_{L,x}})$ the spectral
projection corresponding to the interval $[a,b]$ and the operator
$H_{\Lambda_{L,x}}$. Since we have for any $\lambda \in \RR$ and
$0<\epsilon \leq b-a$
\[
 \arctan \left( \frac{\lambda - a}{\epsilon} \right) - \arctan \left( \frac{\lambda - b}{\epsilon} \right) \geq \frac{\pi}{4} \ \chi_{[a,b]}(\lambda) ,
\]
one obtains an inequality version of Stones formula:
\[
 \langle \delta_x , P_{[a,b]} (H_{\Lambda_{L,x}}) \delta_x \rangle
\leq \frac{4}{\pi} \int_{[a,b]} \im \left\{ G_{\Lambda_{L,x}} (E+ \i \epsilon ; x,x) \right\} \drm E \quad \forall \, \epsilon \in (0, b-a] .
\]
Using triangle inequality, $\lvert \im z\rvert \leq \lvert z\rvert$
for $z \in \CC$, Fubini's theorem, $\lvert G_{\Lambda_{L,x}} (E+\i
\epsilon ; x,x) \rvert^{1-s} \leq \dist (\sigma(H_{\Lambda_{L,x}}) ,
E+i \epsilon)^{s-1} \leq \epsilon^{s-1}$ and hypothesis  (ii) we
obtain for all $\epsilon \in (0,b-a]$
\begin{align*}
\EE \bigl\{ \Tr P_{[a,b]}(H_{\Lambda_{L,x}}) \bigr\} & \leq \EE \Bigl\{ \sum_{x \in \Lambda_{L,x}} \frac{4}{\pi} \int_{[a,b]} \im \left\{ G_{\Lambda_{L,x}} (E+\i \epsilon ; x,x) \right\} \drm E  \Bigr\} \\
&  \leq  \frac{\epsilon^{s-1}}{\pi / 4}  \sum_{x \in \Lambda_{L,x}} \int_{[a,b]} \EE \Bigl\{ \bigl|  G_{\Lambda_{L,x}} (E+\i \epsilon ; x,x)  \bigr|^{s} \Bigr\} \drm E   \\
& \leq 4\pi^{-1}\epsilon^{s-1}  \lvert \Lambda_{L,x} \rvert \, \lvert b-a \rvert C' .
\end{align*}
We minimize the right hand side by choosing $\epsilon = b-a$ and obtain for all $[a,b] \subset I$ with $0 < b-a \leq 1$ the Wegner estimate
\begin{equation} \label{eq:wegner}
\EE \bigl\{ \Tr P_{[a,b]}(H_{\Lambda_{L,x}}) \bigr\}
\leq 4\pi^{-1} C' \lvert b-a \rvert^{s}  \lvert \Lambda_{L,x} \rvert
=: C_{\rm W} \lvert b-a \rvert^{s}  \lvert \Lambda_{L,x} \rvert .
\end{equation}
Now we want to estimate the probability of the event $B_{\rm res} := \{\omega \in \Omega :  I \cap I_{\omega,x}(\delta) \cap I_{\omega,y}(\delta) \not = \emptyset \}$
that there are ``resonant'' energies for the two box Hamiltonians $H_{\Lambda_{L,x}}$ and $H_{\Lambda_{L,y}}$.
For this purpose we denote by $\Lambda_{L,x}'$ the set of all lattice sites $k \in \ZZ^d$
whose coupling constant $\omega_k$ influences the potential in $\Lambda_{L,x}$,
i.\,e. $\Lambda_{L,x}' = \cup_{x \in \Lambda_{L,x}} \{k \in \ZZ^d : u(x-k) \not = 0)\}$.
Notice that the expectation in Ineq. \eqref{eq:wegner} may therefore be replaced by $\EE_{\Lambda_{L,x}'}$.
Moreover, since $\lvert x-y \rvert_\infty \geq 2L + n$, the operator $H_{\Lambda_{L,y}}$ and hence the interval $I_{\omega , y} (\delta)$ is independent of $\omega_k$, $k \in \Lambda_{L,x}'$.
We use the product structure of the measure $\PP$, Chebyshev's inequality, and estimate \eqref{eq:wegner} to obtain
\begin{align}
\nonumber
\PP_{\Lambda_{L,x}'} \{B_{\rm res} \}& \leq \sum_{i=1}^{\lvert \Lambda_{L,y} \rvert}
\PP_{\Lambda_{L,x}'} \bigl\{ \omega \in \Omega : \Tr \bigl( P_{I \cap [E_{\omega,y}^i-2\delta , E_{\omega,y}^i + 2\delta ]} (H_{\Lambda_{L,x}}) \bigr) \geq 1 \bigr\} \\
& \leq \sum_{i=1}^{\lvert \Lambda_{L,y} \rvert}
\EE_{\Lambda_{L,x}'} \bigl\{ \Tr \bigl( P_{I \cap [E_{\omega,y}^i-2\delta , E_{\omega,y}^i + 2\delta ]}  (H_{\Lambda_{L,x}}) \bigr) \bigr\} \nonumber \\
&\leq \lvert \Lambda_{L,y} \rvert  C_{\rm W} (4\delta)^{s}\lvert \Lambda_{L,x} \rvert. \label{eq:wegner-application}
\end{align}
Notice that $4\delta \leq 1$, since $L \geq 8 \ln 8$. Consider now an $\omega \not \in \tilde B_x \cup \tilde B_y$. Recall that \eqref{eq:claim} tells us that $\Delta_\omega^x \subset  I_{\omega,x}(\delta)$
and $\Delta_\omega^y \subset  I_{\omega,y}(\delta)$. If additionally $\omega \not \in B_{\rm  res}$ then no $E \in I$ can be in
$\Delta_\omega^x $ and $\Delta_\omega^y$ simultaneously. Hence for each $E \in I$ either  $\Lambda_{L,x}$ or $\Lambda_{L,y}$
is $(\mu/8,E)$-regular. A contraposition gives us
\begin{align*}
\PP \bigl\{&\text{$\exists \, E \in I$,
$\Lambda_{L,x}$ and $\Lambda_{L,y}$ is $(\mu/8,E)$-singular} \bigr\} \\
&\le \PP \{\tilde B_x\} +\PP \{\tilde B_y \} + \PP \{B_{\rm res} \}
\\ &\leq 2\lvert \Lambda_{L,x} \rvert \, \lvert I \rvert C
\euler^{-\mu L /8} + \lvert \Lambda_{L,y} \rvert  C_{\rm W}
(4\delta)^{s}\lvert \Lambda_{L,x} \rvert,
\end{align*}
from which the result follows.
\end{proof}
In the proof of Proposition \ref{prop:replace-msa} its Hypothesis (ii)
was only used to obtain a Wegner estimate, {\it i.e.} Eq.~\eqref{eq:wegner}.
Hence, if we know that a Wegner estimate holds
for some other reason, e.g.~from \cite{Veselic2010},
 we can relinquish the Hypothesis (ii) and
skip the corresponding argument in the proof of Proposition
\ref{prop:replace-msa}. Specifically, the following assertion holds
true:
\begin{proposition} \label{prop:replace-msa-Wegner}
Let $I \subset \RR$ be a bounded interval and $s \in (0,1)$. Assume the following two statements:
\begin{enumerate}[(i)]
\item There are constants $C,\mu \in (0,\infty)$ and $L_0 \in \NN_0$ such that
\[\EE \bigl\{\lvert G_{\Lambda_{L,k}} (E;x,y) \rvert^{s} \bigr\}
\leq C \euler^{-\mu \lvert x-y \rvert_\infty}\]
for all $k \in \ZZ^d$, $L \in \NN$, $x,y \in \Lambda_{L,k}$ with
$\lvert x-y \rvert_\infty \geq L_0$, and all $E \in I$.
\item There are constants $C_{\rm W}\in (0,\infty)$, $ \beta \in (0,1]$, and $D \in \NN$
such that
\[\PP\bigl\{ \sigma(H_{\Lambda_{L,0}} ) \cap
[a,b]\not=\emptyset\bigr\} \leq C_{\rm W}{\lvert b-a\rvert}^\beta \,
L^D\]
for all $L \in \NN$ and all $[a,b]\subset I$.
\end{enumerate}
Then we have for all $L \geq \max\{8 \ln (2)/\mu , L_0 , -(8/5\mu)\ln (\lvert I \rvert / 2)\}$ and all $x,y \in \ZZ$ with $\lvert x-y\rvert_\infty \geq 2L+\diam \Theta + 1$ that
\begin{multline*}
 \PP \{\forall \, E \in I \text{ either $\Lambda_{L,x}$ or $\Lambda_{L,y}$ is
$(\mu/8,E)$-regular} \}  \\ \geq 1- 8(2L+1)^d\rvert(C \, \lvert I \rvert  + C_{\rm W}L^D  ) \euler^{-\mu \beta L/8} .
\end{multline*}
\end{proposition}
\begin{proof}
We proceed as in the proof of Proposition \ref{prop:replace-msa}, but replace Ineq.{} \eqref{eq:wegner-application} by
\begin{align*}
\PP_{\Lambda_{L,x}'} \{B_{\rm res}\}&\leq \sum_{i=1}^{\lvert \Lambda_{L,y} \rvert}
\PP_{\Lambda_{L,x}'} \bigl\{  I \cap \sigma(H_{\Lambda_{L,x}}) \cap [E_{\omega,y}^i-2\delta , E_{\omega,y}^i + 2\delta ] \neq \emptyset\bigr\} \\
 &\leq \lvert \Lambda_{L,y} \rvert  C_{\rm W} (4\delta)^{\beta} L^D
\end{align*}
to obtain the desired bound.
\end{proof}
\begin{remark} \label{remark:assii}
Note that the conclusions of Proposition \ref{prop:replace-msa} and \ref{prop:replace-msa-Wegner} tell us that the probabilities of $\{\forall \, E \in I \text{ either $\Lambda_{L,x}$ or $\Lambda_{L,y}$ is $(\mu/8,E)$-regular} \}$ tend to one exponentially fast as $L$ tends to infinity. In particular, for any $p>0$ there is some $\tilde L \in \NN$ such that for all $L \ge \tilde L$:
\[
 \PP \{\forall \, E \in I \text{ either $\Lambda_{L,x}$ or $\Lambda_{L,y}$ is $(m,E)$-regular} \}\ge 1- L^{-2p}.
\]
\end{remark}
We will yield exponential localization from the estimates provided
by Proposition \ref{prop:replace-msa} /
\ref{prop:replace-msa-Wegner} using Theorem 2.3 in
\cite{DreifusK1989}. More precisely we need a slight extension of
the result, which can be proven with the same arguments as the
original result. What matters for the proof of
Theorem~\ref{thm:vDK-2.3} is that there is an $l_0 \in \NN$ such
that potential values at different lattice sites are independent if
their distance is larger or equal $l_0$.

\begin{theorem}[\cite{DreifusK1989}] \label{thm:vDK-2.3}
Let $I\subset \RR$ be an interval and let $p>d$, $ L_0>1$, $\alpha \in (1,2p/d)$ and $m>0$.
Set $L_{k} =L_{k-1}^\alpha$, for $k \in \NN$. Suppose that for any $k \in \NN_0$
\[
 \PP \{\forall \, E \in I \text{ either $\Lambda_{L_k,x}$ or $\Lambda_{L_k,y}$ is $(m,E)$-regular} \}\ge 1- L_k^{-2p}
\]
for any $x,y \in \ZZ^d$ with $\lvert x-y\rvert_\infty \geq 2L_k + \diam \Theta + 1$.
Then $H_\omega$ exhibits exponential localization in $I$ for almost all $\omega \in \Omega$.
\end{theorem}

\begin{proof}[Proof of Theorem \ref{thm:result2}]
We assume first that $I$ is a \emph{bounded} interval.
 Fix $E \in I$, $k \in \ZZ^d$ and $L \in \NN$.
By the assumption  of the theorem, Hypothesis (ii)  of \ref{prop:replace-msa} and thus a Wegner estimate hold.
Therefore, for any $L \in \NN$ and any $k \in \ZZ^d$ the probability of finding an eigenvalue of $H_{\Lambda_{L,k}}$ in $[a,b] \subset I$
shrinks to zero as $b-a\to 0$. Hence $E \in I$ is not an eigenvalue of $H_{\Lambda_{L,k}}$
and the resolvent of $H_{\Lambda_{L,k}}$  at $E$ is well defined for all $\omega\in \Omega_I$,
where $\Omega_I$  is a set of full measure.
Lebesgues Theorem now gives
\begin{align}
 C {\rm e}^{-\mu \lvert x-y \rvert_\infty} &\geq \lim_{\epsilon \to 0} \mathbb{E} \bigl\{\lvert G_{\Lambda_{L,k}} (E + \i \epsilon;x,y)\rvert^{s}\bigr\} \nonumber \\
&= \lim_{\epsilon \to 0} \int_{\Omega_I}  \lvert G_{\Lambda_{L,k}} (E + \i \epsilon;x,y)\rvert^{s} \PP (\drm \omega) \nonumber \\
&= \mathbb{E} \bigl\{\lvert G_{\Lambda_{L,k}} (E;x,y)\rvert^{s}\bigr\} \label{eq:bullet} .
\end{align}
For sets of measure zero, the integrand in \eqref{eq:bullet} may not
be defined. However, for the bounds on the expectation value this is
irrelevant. Hence the assumptions of
Proposition~\ref{prop:replace-msa} are satisfied. Combining the
latter with Theorem~\ref{thm:vDK-2.3} and Remark~\ref{remark:assii}
we arrive  to the desired result.

If $I$ is an \emph{unbounded} interval, we can cover it by a countable collection
of bounded intervals. In each of those, exponential localization holds
by the previous arguments for all $\omega$ outside a set of zero measure.
Since the collection of intervals is countable, we have exponential localizaition in $I$ almost surely.
\end{proof}

\begin{proof}[Proof of Theorem~\ref{thm:result3}]
We use Theorem~\ref{thm:result1}  to verify that the hypothesis of
Theorem~\ref{thm:result2}  is satisfied with $I=\RR$. This yields
the desired result.
\end{proof}
\appendix
\section{A non-local apriori bound}

An important step in the proof of exponential decay of fractional moments
is the so called a-priori bound, i.\,e.~a uniform bound on the expectation value
of a fractional power of Green's function elements, which depends in an appropriate
way on the disorder.
It was this step, where the boundary-monotonicity Assumption \ref{ass:monotone}
enters the proof of decay of fractional moments and exponential localization,
as presented in the main body of the paper.

Here in the Appendix we present an alternative a-priori bound
which holds under much milder hypotheses on $u$, see \ref{ass:averagepos} below.
By `milder' we do not mean that this cover the class of models where
\ref{ass:monotone} is satisfied, but rather that it holds generically
in the class of compactly supported single site potentials.

\begin{assumption} \label{ass:averagepos}
\leavevmode \\[-2ex]
\begin{enumerate}
 \item [(B1)] The measure $\mu$ has a density $\rho$ in the Sobolev space $ W^{1,1} (\RR)$.
\item [(B2)]
 The single site potential $u$ satisfies $\overline u := \sum_{k \in \ZZ^d} u(k) \neq 0$.
\end{enumerate}
\end{assumption}
\begin{remark}
Note that without loss of generality (B2) can be replaced by
$\overline u
> 0$, since
\[
V_\omega (x) := \sum_{k \in \ZZ^d} \omega_k u (x-k)= \sum_{k \in \ZZ^d}\big(- \omega_k\big) \, \big(-u (x-k)\big).
\]
\end{remark}
The purpose of this section is to prove
\begin{theorem} \label{theorem:apriori_average}
 Let $\Lambda \subset \ZZ^d$ finite, $s \in (0,1)$ and Assumption~\ref{ass:averagepos} be satisfied. Then we have for all $x,y \in \Lambda$ and $z \in \CC \setminus \RR$
\[
\mathbb{E} \Bigl\{ \bigl\lvert G_\Lambda (z;x,y) \bigr\rvert^s \Bigr\}  \leq
\frac{2^s}{1-s} \bigl(\overline{u}^{-1} \lVert \rho' \rVert_{L^1} C_{\rm W} D \bigr)^s \frac{1}{\lambda^s}
\]
where $D$ and $C_{\rm W}$ are the constants from Eq. \eqref{eq:D} and Lemma~\ref{lemma:monotone}.
\end{theorem}

\begin{remark} This extends Theorem 2.3 of \cite{TautenhahnV-10}.
The drawback of the a-priori bound based on Assumption \ref{ass:averagepos}
is that it is `non-local' in the sense that it requires averaging over
the entire disorder present in the model.
At the moment we are not able to conclude exponential decay of fractional moments
relying in this version of the a-priori bound.
\end{remark}

The proof relies on a special transformation of the random variables
$\omega_k$, $k \in \Lambda_+$, where
$\Lambda_+ = \cup_{k \in \Lambda} \{x \in \ZZ^d \mid u(x-k) \not = 0\}$
denotes the set of lattice sites whose coupling constant influences the
potential in $\Lambda$.
\par
Let $n$ denote the diameter of $\Theta$ with respect to the $\ell^1$-norm, i.\,e. $n := \max_{i,j \in \Theta} \lvert i-j \rvert_1$. For $x,y \in \ZZ^d$ we define $\alpha^{x,y} : \ZZ^d \to \RR^+$ by
\begin{equation} \label{eq:alpha}
 \alpha^{x,y} (k) := \frac{1}{2} \left( {\rm e}^{-c \lvert k-x \rvert_1} + {\rm e}^{-c \lvert k-y \rvert_1} \right) \! \quad \! \text{with} \! \quad \! c:=  \frac{1}{n} \ln \left( 1 + \frac{\overline{u}}{2 \lVert u \rVert_{\ell^1}} \right) .
\end{equation}
Notice that the $\ell^1$-norm of $\alpha^{x,y}$ is independent of $x,y \in \ZZ^d$, i.\,e.
\begin{equation} \label{eq:D}
 D := D(n,\bar u , \lVert u \rVert_{\ell^1}) := \sum_{k \in \ZZ} \lvert \alpha^{x,y} (k) \rvert = \sum_{k \in \ZZ^d} {\rm e}^{-c\lvert k \rvert_1} = \left( \frac{{\rm e}^{c} + 1}{{\rm e}^{c} - 1} \right)^d .
\end{equation}
With the help of the coefficients $\alpha^{x,y} (k)$, $k \in \ZZ^d$, we will define a linear transformation of the variables $\omega_k$, $k \in \Lambda_+$, where $\Lambda_+$ denotes the set of lattice sites whose coupling constants influence the potential in $\Lambda$. Some part of the ``new'' potential will then be given by $W^{x,y} : \ZZ^d \to \RR$,
\begin{equation} \label{eq:Wij}
 W^{x,y} (k) := \sum_{j \in \ZZ^d} \alpha^{x,y} (k) u(k-j),
\end{equation}
where indeed only the values $k \in \Lambda$ are relevant. For our analysis it is important that $W^{x,y}$ is positive and that $W^{x,y} (k) \geq \delta > 0$ for $k \in \{x,y\}$ where $\delta$ is independent of $\Lambda$ and $x,y \in \Lambda$. This is done by
\begin{lemma} \label{lemma:Wij}
 Let Assumption~\ref{ass:averagepos} be satisfied. Then we have for all $x,y,k \in \ZZ^d$
\[
 W^{x,y} (k) \geq \alpha^{x,y} (k) \frac{\overline{u}}{2} > 0 .
\]
In particular, $W^{x,y} (k) \geq \overline{u} / 4$ for $k \in \{x,y\}$.
\end{lemma}
A linear combination with appropriately chosen, exponentiall decaying coefficients,
resp.~a convolution with an exponentially decreasing function is useful also
for other spectral averaging bounds.
See \cite{Veselic-gauss} for an application in the context of Gaussian random potentials in continuum space
and Section 3 in \cite{KostrykinV2006} for abstract criteria, when monotone contributions can be extracted
from a general alloy-type potential.

\begin{proof}
Recall that $n := \max_{i,j \in \Theta} \lvert i-j \rvert_1$ and that we have assumed $0 \in \Theta$. For $k \in \ZZ^d$ let $B_n(k) = \{j \in \ZZ^d : \lvert j-k \rvert_1 \leq n\}$. The triangle inequality gives us for all $k \in \ZZ^d$
\begin{align*}
 M &= \max_{j \in B_n (k)} \bigl\lvert \alpha^{x,y}(k) - \alpha^{x,y} (j) \bigr\rvert \\ &\leq
 \frac{1}{2}\max_{j \in B_n (k)} \bigl\lvert \euler^{-c\lvert k-x \rvert_1} - \euler^{-c\lvert j-x \rvert_1} \bigr\rvert + \frac{1}{2}\max_{j \in B_n (k)} \bigl\lvert \euler^{-c\lvert k-y \rvert_1} - \euler^{-c\lvert j-y \rvert_1} \bigr\rvert .
\end{align*}
Since $\RR \ni t \mapsto {\rm e}^{-ct}$ is a convex and strictly decreasing function, we have for all $k \in \ZZ^d$
\begin{align}
  M &\leq
  \frac{1}{2} \bigl\lvert \euler^{-c\lvert k-x \rvert_1} - \euler^{-c(\lvert k-x \rvert_1 - n)} \bigr\rvert
 +\frac{1}{2} \bigl\lvert \euler^{-c\lvert k-y \rvert_1} - \euler^{-c(\lvert k-y \rvert_1 - n)} \bigr\rvert \nonumber
 \\
&\leq \alpha^{x,y} (k) ({\rm e}^{c n} - 1) . \label{eq:alphamax}
\end{align}
We use Ineq. \eqref{eq:alphamax} and that $u(k-j) = 0$ for $k-j \not \in \Theta$, and obtain the estimate
\begin{align*}
 W^{x,y} (k) &= \sum_{j \in \ZZ^d} \alpha^{x,y} (k) u (k-j) + \sum_{j \in \ZZ^d} \bigl[\alpha^{x,y} (j) - \alpha^{x,y} (k)\bigr] u (k-j) \\
&\geq \alpha^{x,y} (k) \overline u - \sum_{j \in \ZZ^d} \bigl\lvert \alpha^{x,y} (k) - \alpha^{x,y} (j) \bigr\rvert  \bigl\lvert u (k-j) \bigr\rvert \\
& \geq \alpha^{x,y} (k) \overline u - \alpha^{x,y} (k) (\euler^{cn} - 1) \lVert u \rVert_{\ell^1} .
\end{align*}
This implies the statement of the lemma due to the choice of $c$.
\end{proof}
\begin{proof}[Proof of Theorem~\ref{theorem:apriori_average}] Without loss of generality we assume $z \in \CC^-:=\{z \in \CC \mid \Im (z) <0\}$.
 Fix $x,y \in \Lambda$ and recall that  $\Lambda_+$ is the set of lattice sites whose coupling constant influences the potential in $\Lambda$.
We consider the expectation
\[
 E = \EE \Bigl \{ \bigl| G_\Lambda (z;x,y) \bigr|^s \Bigr\} =
\int_{\Omega_{\Lambda_+}} \bigl| \bigl\langle \delta_x , (H_\Lambda - z)^{-1} \delta_y \bigr\rangle \bigr|^s k(\omega_{\Lambda_+}) \drm \omega_{\Lambda_+} ,
\]
where $\Omega_{\Lambda_+} = \times_{k \in \Lambda_+} \RR$, $\omega_{\Lambda_+} = (\omega_k)_{k \in \Lambda_+}$, $k (\omega_{\Lambda_+}) = \prod_{k \in \Lambda_+} \rho (\omega_k)$ and $\drm \omega_{\Lambda_+} = \prod_{k \in \Lambda_+} \drm \omega_k$. Fix $v \in \Lambda_+$. We introduce the change of variables
\[
 \omega_v = \alpha^{x,y} (v) \zeta_v , \quad \text{and} \quad \omega_k = \alpha^{x,y} (k) \zeta_v + \alpha^{x,y} (v) \zeta_k
\]
for $k \in \Lambda_+ \setminus \{v\}$, where $\alpha^{x,y} : \ZZ^d \to \RR^+$ is defined in Eq. \eqref{eq:alpha}. With this transformation we obtain
\begin{align}
 E &= \int_{\Omega_{\Lambda_+}} \bigl| \bigl\langle \delta_x , (-\Delta_\Lambda + \lambda V_\Lambda - z)^{-1} \delta_y \bigr\rangle \bigr|^s k(\omega_{\Lambda_+}) \drm \omega_{\Lambda_+} \nonumber \\
   &= \int_{\Omega_{\Lambda_+}} \bigl|\bigl\langle \delta_x , \bigl(A + \zeta_v \lambda W^{x,y} \bigr)^{-1} \delta_y \bigr\rangle \bigr|^s \tilde k (\zeta_{\Lambda_+})  \drm \zeta_{\Lambda_+} , \label{eq:E}
\end{align}
where $\zeta_{\Lambda_+} = (\zeta_k)_{k \in \Lambda_+}$, $$\tilde k
(\zeta_{\Lambda_+}) = \lvert \alpha^{x,y} (v) \rvert^{\lvert
\Lambda_+ \rvert}\rho (\alpha^{x,y} (v) \zeta_v) \prod_{k \in
\Lambda_+ \setminus \{v\}} \rho (\alpha^{x,y} (k) \zeta_v +
\alpha^{x,y} (v) \zeta_k),$$ $\drm \zeta_{\Lambda_+} = \prod_{k \in
\Lambda_+} \drm \zeta_k$, $A= -\Delta_\Lambda - z + \alpha^{x,y} (0)
\sum_{k \in \Lambda_+ \setminus \{0\}} \zeta_k u(\cdot - k)$ and
$W^{x,y} : \ell^2 (\Lambda) \to \ell^2 (\Lambda)$ is the
multiplication operator with multiplication function given by Eq.
\eqref{eq:Wij}. Notice that $A$ is independent of $\zeta_0$ and
$W^{x,y}$ is positive by Lemma \ref{lemma:Wij}. We use Fubini's
theorem to integrate first with respect to $\zeta_v$. Let $P_x,P_y :
\ell^2 (\Lambda) \to \ell^2 (\Lambda)$ be the orthogonal projection
onto the state $\delta_x$ and $\delta_y$, respectively. The layer
cake representation, see e.\,g. \cite[p. 26]{LiebL2001}, gives us
\begin{align*}
 I &= \int_{\RR} \bigl|\bigl\langle \delta_x , \bigl(A + \zeta_v \lambda W^{x,y} \bigr)^{-1} \delta_y \bigr\rangle \bigr|^s \tilde k (\zeta_{\Lambda_+})  \drm \zeta_{v} \\
&\leq \int_0^\infty \int_\RR \mathbf{1}_{\{\lVert P_x(A + \zeta_v \lambda W^{x,y})^{-1} P_y \rVert^s > t\}} \tilde k (\zeta_{\Lambda_+}) \drm \zeta_v \drm t .
\end{align*}
We decompose the integration domain into $[0,\kappa]$ and $[\kappa,\infty)$ with $\kappa > 0$. In the first integral we estimate the characteristic function one. In the second integral we estimate $\tilde k (\zeta_{\Lambda_+}) \leq \sup_{\zeta_v \in \RR} \tilde k (\zeta_{\Lambda_+})$ and then use Lemma \ref{lemma:monotone}. This gives
\begin{equation} \label{eq:I}
 I \leq \kappa \int_\RR \tilde k (\zeta_{\Lambda_+}) \drm \zeta_v + \frac{C_{\rm W} \lambda^{-1}}{[W^{x,y} (x) W^{x,y} (y)]^{1/2}} \sup_{\zeta_v \in \RR} \tilde k (\zeta_{\Lambda_+}) \int_{\kappa}^\infty \frac{1}{t^{1/s}} \drm t .
\end{equation}
We use $\int_{\kappa}^\infty t^{-1/s} \drm t = [s/(1-s)] \kappa^{(s-1)/s}$, the fact that $\tilde k$ is a probability density and the estimate $\sup_{x \in \RR} g(x) \leq \frac{1}{2} \int_\RR \lvert g'(x)\rvert \drm x$ for $g \in W^{1,1} (\RR)$, and obtain from Ineq. \eqref{eq:E} and Ineq. \eqref{eq:I}
\[
 E \leq \kappa + \frac{C_{\rm W} \lambda^{-1} \frac{s}{1-s} \kappa^{\frac{s-1}{s}}{}}{[W^{x,y} (x) W^{x,y} (y)]^{1/2}}  \frac{1}{2} \int_{\Omega_{\Lambda_+}} \biggl| \frac{\partial \tilde k (\zeta_{\Lambda_+})}{\partial \zeta_v} \biggr| \drm \zeta_{\Lambda_+} .
\]
For the partial derivative we calculate
\[
 \frac{\partial \tilde k (\zeta_{\Lambda_+})}{\partial \zeta_0} = \lvert \alpha^{i,j} (v) \rvert^{\lvert \Lambda_+ \rvert}
\sum_{l \in \Lambda_+} \alpha^{i,j} (l) \rho' (\omega_l) \prod_{\genfrac{}{}{0pt}{2}{k \in \Lambda_+}{k \not = l}} \rho (\omega_k) ,
\]
which gives (while substituting into original coordinates)
\begin{align*}
E & \leq \kappa + \frac{C_{\rm W} \lambda^{-1} \frac{s}{1-s} \kappa^{\frac{s-1}{s}}{}}{[W^{x,y} (x) W^{x,y} (y)]^{1/2}}  \frac{1}{2}
\sum_{l \in \Lambda_+} \lvert \alpha^{i,j} (l) \rvert \int_{\Omega_{\Lambda_+}}
\lvert \rho' (\omega_l) \rvert \prod_{\genfrac{}{}{0pt}{2}{k \in \Lambda_+}{k \not = l}} \lvert \rho (\omega_k) \rvert
\drm \omega_{\Lambda_+} \\
&\leq \kappa + \frac{C_{\rm W} \lambda^{-1} \frac{s}{1-s} \kappa^{\frac{s-1}{s}}{}}{[W^{x,y} (x) W^{x,y} (y)]^{1/2}}  \frac{1}{2}
D \lVert \rho' \rVert_{L^1} \leq \kappa + \frac{C_{\rm W} \lambda^{-1} \frac{s}{1-s} \kappa^{\frac{s-1}{s}}{}}{\overline{u} / 2}
D \lVert \rho' \rVert_{L^1},
\end{align*}
where $D$ is the constant from Eq. \eqref{eq:D} and where we have used that $W^{x,y}(x)$ and $W^{x,y}(y)$ are bounded from below by $\overline{u}/4$ by Lemma \ref{lemma:Wij}. If we choose $\kappa = (\lVert \rho \rVert_{L^1} C_{\rm W}$ $\lambda^{-1} 2 D / \overline{u})^s$ we obtain the statement of the theorem.
\end{proof}
\subsection*{Acknowledgment}
Part of this work was done while the authors were attending a
mini-workshop at the Mathematisches Forschungsinstitut Oberwolfach.
A.E. has been partially supported by NSF grant DMS--0907165. M.T.
and I.V. have been partially supported by DFG grants.
%
%
%
%
\providecommand{\bysame}{\leavevmode\hbox to3em{\hrulefill}\thinspace}
\providecommand{\MR}{\relax\ifhmode\unskip\space\fi MR }
\providecommand{\MRhref}[2]{%
  \href{http://www.ams.org/mathscinet-getitem?mr=#1}{#2}
}
\providecommand{\href}[2]{#2}

\end{document}